\documentclass[11pt,a4paper]{article}
\usepackage{amsmath, amsfonts, amssymb}
\usepackage{a4wide}
\usepackage{parskip}
\usepackage{enumitem}
\usepackage{xcolor}

\usepackage{hyperref}
\usepackage{booktabs}
\usepackage{caption}
\usepackage{siunitx}
\usepackage{tabularx}
\usepackage{graphicx}
\usepackage{adjustbox}
\usepackage{multirow}
\usepackage{amsthm}
\usepackage{bm}
\usepackage{lscape}
\usepackage{float}
\usepackage{comment}
\usepackage{nicematrix}
\usepackage[noadjust]{cite}

\def\qed{\hfill {$\square$}\goodbreak \medskip}

\newtheorem{theorem}{Theorem}[section]
\newtheorem{lemma}[theorem]{Lemma}
\theoremstyle{definition}
\newtheorem{definition}[theorem]{Definition}
\newtheorem{example}[theorem]{Example}
\newtheorem{proposition}[theorem]{Proposition}
\theoremstyle{remark}
\newtheorem{remark}[theorem]{Remark}

\newtheorem{corollary}[theorem]{Corollary}

\numberwithin{equation}{section}

\newcommand{\Char}{\textnormal{char}}
\newcommand{\Tr}{\textnormal{Tr}}
\newcommand{\Trqr}{\textnormal{Tr}_{q^r/q}}
\newcommand{\Nrqr}{\textnormal{N}_{q^r/q}}
\newcommand{\Nrqm}{\textnormal{N}_{q^d/q}}
\newcommand{\Span}{\textnormal{Span}}
\newcommand{\ars}{\textnormal{ARS}}
\newcommand{\atgrs}{\textnormal{ATGRS}}
\newcommand{\atrs}{\textnormal{ATRS}}
\newcommand{\agrs}{\textnormal{AGRS}}

\newcommand{\N}{\textnormal{N}}
\newcommand{\twist}{\mathcal{V}_{k, \bm{t}, \bm{h}, \bm{\eta}}}
\newcommand{\wt}{\textnormal{wt}}
\usepackage{tikz,xcolor,hyperref}

\definecolor{lime}{HTML}{A6CE39}
\DeclareRobustCommand{\orcidicon}{%
	\begin{tikzpicture}
		\draw[lime, fill=lime] (0,0) 
		circle [radius=0.16] 
		node[white] {{\fontfamily{qag}\selectfont \tiny ID}};
		\draw[white, fill=white] (-0.0625,0.095) 
		circle [radius=0.007];
	\end{tikzpicture}
	\hspace{-2mm}
}

\foreach \x in {A, ..., Z}{%
	\expandafter\xdef\csname orcid\x\endcsname{\noexpand\href{https://orcid.org/\csname orcidauthor\x\endcsname}{\noexpand\orcidicon}}
}

\begin{document}
	\date{}
	{\vspace{0.01in}
        \title{New Constructions of Additive MDS TRS Codes\\ }

        \author{{\bf Anuj Kumar Bhagat\footnote{email: {\tt anujkumarbhagat632@gmail.com}}\orcidA{},\;\bf Ankit Yadav\footnote{email: {\tt ankityadav10102000@gmail.com}}\orcidH{}, and \bf Ritumoni Sarma\footnote{    email: {\tt ritumoni407@gmail.com}}\orcidR{}} \\ $^{\ast \dagger\ddagger}$Department of Mathematics,\\ Indian Institute of Technology Delhi,\\Hauz Khas, New Delhi-110016, India \medskip \\}
\maketitle
\begin{abstract}
Additive codes over finite fields generalize linear codes, and additive MDS codes provide a natural extension of linear MDS codes. In this article, we study additive twisted Reed--Solomon (TRS) codes and obtain new constructions of additive MDS codes. First, for additive TRS codes with twist $t=2$ and an arbitrary hook, we establish necessary and sufficient conditions for the codes to be additive MDS, thereby generalizing the results of \cite[Section 3]{JiayuMa2026}. In particular, we show that the existence of an additive MDS TRS code with $t=2$ and hook $h=0$ yields codes of larger lengths than those obtained for $t=2$ and $h=k-1$ in \cite{JiayuMa2026}. Next, we consider additive TRS codes with twist vector $\bm{t}=(1,2)$ and hook vector $\bm{h}=(0,0)$, and derive necessary and sufficient conditions for them to be additive MDS. We further establish the existence of such codes. Using the Schur square technique, we obtain mild conditions under which the constructed families are inequivalent to additive Reed--Solomon (RS) codes. Finally, we determine parity-check matrices for both families of additive MDS codes considered in this article.

    \medskip
		\noindent \textit{Keywords:} Additive codes; Additive Reed--Solomon codes; Twisted Reed--Solomon codes; MDS codes; Schur product.
			\medskip
\end{abstract}
\section{Introduction}\label{Section 1}
\subsection{Background}
    For an $[n,k,d]$-linear code, the Singleton bound states that $d\le n-k+1.$ Linear codes attaining this bound are called maximum distance separable (MDS) codes. Thus, MDS codes attain the largest possible Hamming distance for a given length and dimension, and consequently provide optimal error-detection and error-correction capabilities within these parameters. Owing to their optimality, MDS codes have been extensively studied from several perspectives, including their classification up to equivalence \cite{kokkala2015classification}, weight distributions \cite{alderson2020weights}, \cite{ezerman2010weights}, covering radii \cite{bartoli2014covering}, LCD properties \cite{carlet2018euclidean}, and applications in cryptography \cite{Niederreiter1986Knapsack}. Moreover, MDS codes are closely related to finite geometry and combinatorial designs, providing important connections between coding theory and these areas \cite{hirschfeld1998projective}, \cite{macwilliams1977theory}.

    A fundamental and extensively studied family of MDS codes is given by the Generalized Reed--Solomon (GRS) codes \cite{Roth_2006}. A GRS code over $\mathbb{F}_q$ has length at most $q$, while its extended version yields MDS codes of length at most $q+1$. This naturally leads to the question of how large the length of an MDS code over $\mathbb{F}_q$ can be. The celebrated MDS conjecture asserts, broadly speaking, that a non-trivial MDS code over $\mathbb{F}_q$ has length at most $q+1$, apart from certain exceptional parameter cases in which length $q+2$ may occur \cite{segre1955curve}. Significant progress has been made toward the conjecture; in particular, it is known to hold for prime fields \cite{ball2012large}.
    
    Although GRS codes constitute the most prominent family of MDS codes, not every MDS code is equivalent to a GRS code. The first known construction of MDS codes not equivalent to GRS codes was given by Roth and Lempel in 1989 \cite{RothLempel1989}. Subsequently, motivated by the work of Sheekey \cite{JohnSheekey2016}, Beelen et al. introduced the family of Twisted Reed--Solomon (TRS) codes \cite{beelen2017twisted,beelen2022twisted}. Certain subfamilies of TRS codes were shown to be MDS and not equivalent to GRS codes. The introduction of twisted terms into the polynomial representation of Reed--Solomon codes thus provides a flexible mechanism for constructing new MDS codes beyond the classical GRS framework. This approach has subsequently led to numerous constructions of MDS codes inequivalent to GRS codes; see, for example, \cite{abdukhalikov2026two, zhi2025new,jin2024new,li2025non,wu2024more,chen2023many,liu2026column, bhagat2025row, liu2026constructions}.
    
    For $r\geq 2$, an $\mathbb{F}_{q}$-linear subspace of $\mathbb{F}_{q^r}^{n}$ is called an additive code of length $n$ over $\mathbb{F}_{q^r}$. Such codes are also referred to as $\mathbb{F}_{q}$-linear $\mathbb{F}_{q^r}$-codes in the literature. Here, the term ``additive'' is used in the standard sense of $\mathbb{F}_{q}$-linearity and does not merely refer to closure under addition. Every $\mathbb{F}_{q^r}$-linear code is naturally an additive code, whereas the converse need not hold. Thus, additive codes provide a broader class of codes than linear codes over the same alphabet and, in particular, allow the construction of codes with parameter sets that may not be attainable within the linear framework.
    Additive codes have also found important applications in quantum information \cite{Calderbank1998quantum,Ketkar2006nonbinary}, computer memory systems \cite{chen1984error,chen1991fault,chen1992symbol}, deep-space communication \cite{hattori1998subspace}, and secret sharing \cite{Kim2017secret}.

    Calderbank et al. \cite{Calderbank1998quantum} first introduced the notion of additive codes over $\mathbb{F}_{4}$ and used additive self-orthogonal codes to construct binary quantum codes. Later, the theory of additive codes was generalized to arbitrary fields in \cite{rains1999nonbinary, Bierbrauer2000quantum}. Additive cyclic codes over $\mathbb{F}_{q^r}$ were studied by Huffman in \cite{huffman2010cyclic}, while Huffman in \cite{Huffman2013} investigated additive codes and their duals with respect to trace inner products, including the trace-Euclidean and trace-Hermitian inner products.
    The notion of MDS codes extends naturally to the additive setting. In particular, Huffman \cite{Huffman2013} established a Singleton bound for additive codes, and an additive code attaining this bound is called an additive MDS code. Their structure has recently been studied from a geometric viewpoint: Ball et al. \cite{ball2023additivemds} established a correspondence between additive MDS codes and certain geometric objects known as pseudo-arcs. More recently, minimal additive codes were introduced and investigated from a geometric perspective in \cite{alfarano2026minimal}. Further developments on additive codes can be found in \cite{shi2022cyclicacd,shi2023acd,gyan,gao2021cyclic,sharma2017cyclic,guan2023,abdukhalikov2026quasi,ball2025griesmer,abdukhalikov2026equivalence, Yadav2026} and the references therein.
    

\subsection{Motivation and Objectives}
Yadav and Sharma \cite{sharma2024mds} first introduced the concept of additive generalized Reed--Solomon (AGRS) codes and established necessary and sufficient conditions for these codes to be additive MDS. They also investigated additive twisted generalized Reed--Solomon (ATGRS) codes with a single twist $t=1$ and an arbitrary hook, and obtained a family of additive MDS codes which are not equivalent to AGRS. Subsequently, in \cite{JiayuMa2026}, the authors studied a new class of additive codes arising from additive twisted Reed-Solomon (ATRS) codes with twist $t=2$ and hook $h=k-1$. They obtained additive MDS codes that are not equivalent to AGRS codes.

Motivated by the above two works on additive TGRS codes, we pose the following questions:
\begin{itemize}
    \item Can the study in \cite{JiayuMa2026} be extended to arbitrary hooks? If so, would this lead to the existence of additive MDS codes of larger lengths that are not equivalent to additive RS codes? 
    \item Moreover, additive TRS codes with a double twist have not yet been studied. Can we construct additive MDS codes from doubly twisted additive TRS codes that are inequivalent to additive RS codes?
    \item In \cite{JiayuMa2026}, the parity-check matrix obtained for the proposed additive TRS codes appears to be incorrect. Can we establish a general method for determining the parity-check matrix of additive TRS codes?
\end{itemize}
\subsection{Our Contribution}
First, we study additive TRS codes with a single twist $t=2$ and arbitrary hook $0\leq h\leq k-1$, and establish necessary and sufficient conditions for these codes to be additive MDS under mild assumptions on the evaluation points. For the case $h=0$, by determining the number of elements of the extension field $\mathbb{F}_{q^r}/\mathbb{F}_q$ having exactly $r$ distinct conjugates and trace zero, as well as the number of elements having exactly $r$ distinct conjugates and norm $1$, we establish the existence of additive MDS codes that are not equivalent to additive RS codes of lengths larger than those obtained for $h=k-1$ in \cite{JiayuMa2026}. Additionally, following the study of TGRS codes with twists $\bm{t}=(1,2)$ and hooks $\bm{h}=(0,0)$ in \cite{yang2025TGRS}, we investigate the corresponding additive TRS codes with the same twists and hooks, and derive criteria for these codes to be additive MDS under certain assumptions.

Furthermore, we explicitly determine the parity-check matrices of both the above-constructed families. By analyzing the Schur squares of these codes, we obtain several families of additive MDS codes arising from additive TRS codes that are inequivalent to additive RS codes.

The article is organized as follows. Section \ref{Section 2} presents the basic terminology and necessary preliminaries. In Section \ref{Section 3}, we introduce the first family of additive TRS codes with single twist $t=2$ and arbitrary hook $0\leq h \leq k-1$ and identify several classes of additive MDS codes that are inequivalent to additive RS codes. In Section \ref{Section 4}, we present the second family of additive TRS codes with double twists $\bm{t}=(1,2)$ and hooks $\bm{h}=(0,0)$ and identify several classes of additive MDS codes inequivalent to additive RS codes. In Section \ref{Section 5}, we determine the parity-check matrices of the additive TRS codes considered in Sections \ref{Section 3} and \ref{Section 4}. Finally, Section \ref{Section 6} concludes the article.

\section{Preliminaries}\label{Section 2}
Throughout this article, let $n$ denote a positive integer, $r\ge 2$ an integer, and $q$ a prime power. Let $\mathbb{F}_q$ be the finite field of order $q$, and let $\mathbb{F}_{q^r}$ be the field extension of $\mathbb{F}_q$ of degree $r$. The \textit{Frobenius automorphism} $\sigma_q$ of $\mathbb{F}_{q^r}$ over $\mathbb{F}_{q}$ is defined by
\begin{align*}
    \sigma_q: \mathbb{F}_{q^r}&\to \mathbb{F}_{q^r},\\
    \alpha&\mapsto\alpha^q,  
\end{align*}
for every $\alpha\in\mathbb{F}_{q^r}.$ Moreover, the Galois group of the extension $\mathbb{F}_{q^r}/\mathbb{F}_q$ is cyclic group of order $r$ and is generated by $\sigma_q$; that is, 
\begin{equation*}
    \text{Gal}(\mathbb{F}_{q^r}/\mathbb{F}_q)=\langle \sigma_q \rangle=\{\sigma_q^j\ :\ 0\le j\le r-1 \}.
\end{equation*}
The \textit{conjugates} of an element $\alpha\in\mathbb{F}_{q^r}$ over $\mathbb{F}_q$ are the elements $\alpha,\ \sigma_q(\alpha),\ \sigma_q^2(\alpha),\ \ldots,\ \sigma_q^{r-1}(\alpha).$ These $r$ elements are not necessarily distinct. The set $[\alpha]:=\{\sigma_q^j(\alpha):0\le j\le r-1\}$ is called the \emph{conjugacy class} (or \emph{Frobenius orbit}) of $\alpha$ over $\mathbb{F}_q$. Throughout this article, we fix a subset
\[
\mathcal{B}_{q,r}\subseteq\mathbb{F}_{q^r}
\]
consisting of one representative from each conjugacy class of elements having exactly $r$ distinct conjugates over $\mathbb{F}_q$. It follows from \cite{lidl1997book} that
\[
|\mathcal{B}_{q,r}|=\frac{1}{r}\sum_{d|r} \mu(d) q^{r/d}.
\]


The \textit{trace} and the \textit{norm} maps from $\mathbb{F}_{q^r}$ to $\mathbb{F}_q$ are defined, respectively, by
\[
\begin{aligned}
    \Trqr:\mathbb{F}_{q^r}
    &\longrightarrow \mathbb{F}_q,\\
    \alpha
    &\longmapsto
    \sum_{j=0}^{r-1}\sigma_q^j(\alpha)
    =\alpha+\alpha^q+\alpha^{q^2}+\cdots+\alpha^{q^{r-1}},
\end{aligned}
\]
and
\[
\begin{aligned}
    \Nrqr:\mathbb{F}_{q^r}
    &\longrightarrow \mathbb{F}_q,\\
    \alpha
    &\longmapsto
    \prod_{j=0}^{r-1}\sigma_q^j(\alpha)
    =\alpha\alpha^q\alpha^{q^2}\cdots\alpha^{q^{r-1}}
    =\alpha^{\frac{q^r-1}{q-1}},
\end{aligned}
\]
for every $\alpha\in\mathbb{F}_{q^r}.$

The set $\mathbb{F}_{q^r}^n$ is an $\mathbb{F}_q$-vector space of dimension $nr$ under componentwise addition and scalar multiplication by elements of $\mathbb{F}_q$. Throughout this article, an \emph{additive code} $\mathcal{C}$ of length $n$ over $\mathbb{F}_{q^r}$ is an $\mathbb{F}_q$-subspace of $\mathbb{F}_{q^r}^n$. The \textit{dimension} of $\mathcal{C}$, denoted by $\dim_{\mathbb{F}_q}(\mathcal{C})$, is its dimension as an $\mathbb{F}_q$-vector space. A \emph{generator matrix} of $\mathcal{C}$ is a matrix over $\mathbb{F}_{q^r}$ whose rows form an $\mathbb{F}_q$-basis of $\mathcal{C}$. The \textit{(Hamming) distance} of $\mathcal{C}$, denoted by $d(\mathcal{C})$, is defined by
\[
    d(\mathcal{C})=\min\{\wt(\bm{c}) : \bm{0}\neq \bm{c}\in \mathcal{C}\},
\]
where $\wt(\bm{c})$ denotes the Hamming weight of the vector $\bm{c}$. An additive code of length $n$, dimension $k$, and distance $d$ over $\mathbb{F}_{q^r}$ is referred to as an $(n,q^k,d)$-additive code over $\mathbb{F}_{q^r}$.

\begin{lemma}[\cite{Huffman2013}, Singleton bound for additive codes]
    Any $(n,q^k,d)$-additive code over $\mathbb{F}_{q^r}$ satisfies \
    \[
        d\le n-\left\lceil\frac{k}{r} \right\rceil+1.
    \]
\end{lemma}
An $(n,q^k,d)$-additive code over $\mathbb{F}_{q^r}$ is called an \emph{additive MDS code} if $d=n-\left\lceil\frac{k}{r}\right\rceil+1,$ and an \emph{additive almost MDS code} if $d=n-\left\lceil\frac{k}{r}\right\rceil.$

The \emph{ordinary trace inner product} on $\mathbb{F}_{q^r}^n$ is defined by
\begin{align*}
    \langle \cdot,\cdot\rangle:\mathbb{F}_{q^r}^n\times\mathbb{F}_{q^r}^n &\longrightarrow \mathbb{F}_q,\\
    \langle \bm{u},\bm{v}\rangle
    &=\sum_{i=1}^{n}\Trqr(u_i v_i),
\end{align*}
where $\bm{u}=(u_1,\dots,u_n)$, $\bm{v}=(v_1,\dots,v_n) \in\mathbb{F}_{q^r}^n$. It follows from \cite{Huffman2013} that $\langle\cdot,\cdot\rangle$ is a non-degenerate symmetric $\mathbb{F}_q$-bilinear form on $\mathbb{F}_{q^r}^n$. For an additive code $\mathcal{C}$ over $\mathbb{F}_{q^r}$, its \emph{dual code} is defined by
\[
    \mathcal{C}^{\perp}
    :=\{\bm{v}\in\mathbb{F}_{q^r}^n :
    \langle \bm{c},\bm{v}\rangle=0
    \text{ for all }\bm{c}\in\mathcal{C}\}.
\]
It is well known that $\mathcal{C}^{\perp}$ is again an additive code over $\mathbb{F}_{q^r}$. A generator matrix of $\mathcal{C}^{\perp}$ is called a \emph{parity-check matrix} of $\mathcal{C}$. If $G$ and $H$ are respectively a generator matrix and a parity-check matrix of an additive code $\mathcal{C}$ over $\mathbb{F}_{q^r}$, then
\[
    \Trqr(GH^{T})=O,
\]
where, for a matrix $A=(a_{ij})$ over $\mathbb{F}_{q^r}$, $\Trqr(A)$ denotes the matrix $(\Trqr(a_{ij}))$ over $\mathbb{F}_q$. Furthermore by Proposition 2.4 of \cite{grove2002classical}, we have
\[
    \dim_{\mathbb{F}_q}(\mathcal{C})
    +\dim_{\mathbb{F}_q}(\mathcal{C}^{\perp})
    =nr.
\]
\begin{definition}
Two additive codes $\mathcal{C}$ and $\mathcal{D}$ of length $n$ over $\mathbb{F}_{q^r}$ are said to be \emph{monomially equivalent} if there exists an $n\times n$ monomial matrix $M$ over $\mathbb{F}_{q^r}$ such that
\[
\mathcal{D}=\{\bm{c}M:\bm{c}\in\mathcal{C}\}.
\]
Otherwise, $\mathcal{C}$ and $\mathcal{D}$ are said to be \emph{monomially inequivalent}.
\end{definition}
The additive analogue of Generalized Reed--Solomon (GRS) codes was introduced in \cite{sharma2024mds}.

Let
\[
\mathbb{F}_q[x]_{<k}:=\{f(x)\in\mathbb{F}_q[x]:\deg(f)<k\}\cup\{0\}.
\]

\begin{definition}[\cite{sharma2024mds}]
Let $n$, $k$, and $r\ge 2$ be integers satisfying $1\le k\le nr$. Let $\bm{\alpha}=(\alpha_1,\alpha_2,\dots,\alpha_n)\in\mathbb{F}_{q^r}^n\setminus\mathbb{F}_q^n$ be such that $\alpha_i\neq\alpha_j$ whenever $i\neq j$, and let $\bm{v}=(v_1,v_2,\dots,v_n)\in(\mathbb{F}_{q^r}^{*})^n.$
The \emph{additive Generalized Reed--Solomon (AGRS) code} associated with $(\bm{\alpha},\bm{v})$ is defined by
\[
\agrs_{n,k}(\bm{\alpha},\bm{v})
:=
\left\{
\bigl(v_1f(\alpha_1),v_2f(\alpha_2),\dots,v_nf(\alpha_n)\bigr)
:\,
f(x)\in\mathbb{F}_q[x]_{<k}
\right\}.
\]
\end{definition}

When $\bm{v}=\bm{1}:=(1,1,\dots,1)$, the corresponding AGRS code is called an \emph{additive Reed--Solomon (ARS) code} and is denoted by $\ars_{n,k}(\bm{\alpha})$.

The code $\agrs_{n,k}(\bm{\alpha},\bm{v})$ is an $\mathbb{F}_q$-linear code of length $n$ over $\mathbb{F}_{q^r}$ and hence an additive code over $\mathbb{F}_{q^r}$. In contrast to classical GRS codes, an AGRS code need not be an additive MDS code. Nevertheless, \cite{sharma2024mds} proved that if the entries of $\bm{\alpha}=(\alpha_1, \alpha_2, \dots, \alpha_n)$ belong to $\mathcal{B}_{q,r}$, then $\agrs_{n,k}(\bm{\alpha},\bm{v})$ is an additive MDS code over $\mathbb{F}_{q^r}$.

Analogous to linear twisted Generalized Reed--Solomon (TGRS) codes, one defines additive twisted Generalized Reed--Solomon (ATGRS) codes using twisted polynomials.
\begin{definition}[\cite{sharma2024mds}]
Let $\ell$, $k$, and $r$ be positive integers satisfying $\ell\le k<nr$. Let $\bm{t}=(t_1,t_2,\dots,t_\ell), 1\le t_i\le nr-k, 1\le i\le \ell, \bm{h}=(h_1,h_2,\dots,h_\ell), 0\le h_i<k, 1\le i\le \ell,$
and let $\bm{\eta}=(\eta_1,\eta_2,\dots,\eta_\ell)\in\mathbb{F}_q^\ell.$
The set of \emph{$(k,\bm{t},\bm{h},\bm{\eta})$-twisted polynomials} over $\mathbb{F}_q$ is defined by
\[
\twist
:=
\left\{
\sum_{i=0}^{k-1}a_i x^i
+\sum_{j=1}^{\ell}\eta_j a_{h_j}x^{k-1+t_j}
:\;
a_0,a_1,\dots,a_{k-1}\in\mathbb{F}_q
\right\}.
\]
The integer $\ell$ is called the \emph{number of twists}, the vector $\bm{t}$ the \emph{twist vector}, and the vector $\bm{h}$ the \emph{hook vector}.
\end{definition}
It is straightforward to verify that $\twist$ is a $k$-dimensional $\mathbb{F}_q$-subspace of $\mathbb{F}_q[x]$.

Let $
\bm{\alpha}=(\alpha_1,\alpha_2,\dots,\alpha_n)\in\mathbb{F}_{q^r}^n\setminus\mathbb{F}_q^n$
be such that $\alpha_i\neq\alpha_j$ whenever $i\neq j$, and let $\bm{v}=(v_1,v_2,\dots,v_n)\in(\mathbb{F}_{q^r}^{*})^n.$
The \emph{additive Twisted Generalized Reed--Solomon (ATGRS) code} associated with $(\bm{\alpha},\bm{v},\bm{t},\bm{h},\bm{\eta})$ is defined by
\[
\atgrs_{n,k}(\bm{\alpha},\bm{v},\bm{t},\bm{h},\bm{\eta})
:=
\left\{
\bigl(v_1f(\alpha_1),v_2f(\alpha_2),\dots,v_nf(\alpha_n)\bigr)
:\,
f(x)\in\twist
\right\}.
\]

When $\bm{v}=\bm{1}:=(1,1,\dots,1)$, the corresponding ATGRS code is called an \emph{additive Twisted Reed--Solomon (ATRS) code} and is denoted by $
\atrs_{n,k}(\bm{\alpha},\bm{t},\bm{h},\bm{\eta}).$

Yadav and Sharma \cite{sharma2024mds} considered the case $\ell=1$, $t=1$, and $h\in\{0,1,\dots,k-1\}$, and studied the codes $\atgrs_{n,k}(\bm{\alpha},\bm{v},1,h,\eta)$. They established the necessary conditions for these codes to be additive MDS and proved the existence of additive MDS codes satisfying these conditions. Furthermore, under mild assumptions, they showed that the codes $\atgrs_{n,k}(\bm{\alpha},\bm{v},1,h,\eta)$ are not monomially equivalent to additive Reed--Solomon codes whenever the entries of $\bm{\alpha}$ belong to $\mathcal{B}_{q,r}$.

Subsequently, Ma \emph{et al.} \cite{JiayuMa2026} considered the case $\ell=1$, $t=2$, and $h=k-1$, and established necessary and sufficient conditions for the codes $\atgrs_{n,k}(\bm{\alpha},\bm{v},2,k-1,\eta)$ to be additive MDS. They also proved that, whenever the entries of $\bm{\alpha}$ belong to $\mathcal{B}_{q,r}$, the codes $\atgrs_{n,k}(\bm{\alpha},\bm{v},2,k-1,\eta)$ are not monomially equivalent to additive Reed--Solomon codes. In both works, the principal tool used to establish the inequivalence of additive twisted Reed--Solomon codes and additive Reed--Solomon codes is the Schur square. We therefore recall the notion of the Schur square of an additive code below.

\begin{definition}
    The Schur product of two vectors $\bm{x}=(x_1, x_2, \dots, x_n), \bm{y}=(y_1, y_2, \dots, y_n)\in\mathbb{F}_{q^r}^n$ is defined by 
    \[
    \bm{x}\star \bm{y}:=(x_1y_1, x_2y_2, \dots, x_ny_n).
    \]
    For two additive codes $\mathcal{C}$ and $\mathcal{D}$ of the same length over $\mathbb{F}_{q^r}$, their \textit{Schur product} is an additive code defined by
    \begin{equation}
        \mathcal{C}\star\mathcal{D}:=\Span_{\mathbb{F}_q}\{\bm{c}\star\bm{d}: \bm{c}\in\mathcal{C},\bm{d}\in\mathcal{D}\}.
    \end{equation}
\end{definition}
The Schur product $\mathcal{C}\star\mathcal{C}$ is denoted by $\mathcal{C}^{\star2}$ and is called the \emph{Schur square} of $\mathcal{C}$.
\begin{lemma}\cite[Lemma 6.1]{sharma2024mds}\label{lem: equiv implies dim of schur square equal}
    If $\mathcal{C}$ and $\mathcal{D}$ are two monomially equivalent additive codes over $\mathbb{F}_{q^r}$, then $\dim_{\mathbb{F}_q}(\mathcal{C}^{\star 2})=\dim_{\mathbb{F}_q} (\mathcal{D}^{\star 2}).$
\end{lemma}
The following lemma gives the dimension of the Schur square of an additive Reed--Solomon code.
\begin{lemma}\cite[Theorem 6.1]{sharma2024mds}\label{lem: dim of schur of ars}
    Let $n$, $k$, and $r\ge 2$ be integers satisfying $1\le k\le nr$. Suppose that the entries of $\bm{\alpha}$ belong to $\mathcal{B}_{q,r}$ with $|\bm{\alpha}|=n$. Then
\[
\dim_{\mathbb{F}_q}\left( \ars_{n,k}(\bm{\alpha})^{\star 2}\right)=\min\{nr, 2k-1\}.
\]
\end{lemma}
\begin{remark}
    \begin{enumerate}
        \item[(1)] As in the literature, throughout this article, we assume that the entries of $\bm{\alpha}=(\alpha_1,\alpha_2,\dots,\alpha_n)$ belong to $\mathcal{B}_{q,r}$. Since any permutation of the evaluation points yields a monomially equivalent code, the ordering of the entries of $\bm{\alpha}$ is irrelevant. Accordingly, by abuse of notation, we write $\bm{\alpha}\subseteq\mathcal{B}_{q,r}$.

        \item[(2)] Since $\atgrs_{n,k}(\bm{\alpha},\bm{v},\bm{t},\bm{h},\bm{\eta})$ is monomially equivalent to $\atrs_{n,k}(\bm{\alpha},\bm{t},\bm{h},\bm{\eta})$, throughout this article we restrict our attention to the latter.
    \end{enumerate}
\end{remark}

\section{A Class of Additive MDS Single-Twisted RS Codes}\label{Section 3}
In this section, we consider the single-twist case with $t=2$, $h\in\{0,1,\dots,k-1\}$, and $\eta\in\mathbb{F}_q^*$, and write
\[
\mathcal{C}_{n,k}(\bm{\alpha}, 2, h, \eta)
:=
\atrs_{n,k}(\bm{\alpha},2,h,\eta).
\]
This construction generalizes \cite[Section~3]{JiayuMa2026}, which corresponds to the special case $h=k-1$. We establish new existence results for additive MDS codes, leading to families with larger lengths than those obtained in
\cite{JiayuMa2026}. We further prove that these codes are not monomially equivalent to additive Reed--Solomon codes.
\begin{theorem}\label{thm: AMDS or MDS}
    For $1\le k\le nr,$ the code $\mathcal{C}_{n,k}(\bm{\alpha}, 2, h, \eta)$ over $\mathbb{F}_{q^r}$ is either an additive MDS code or an additive almost MDS code.
\end{theorem}
\begin{proof}
    Since 
    $$\Span_{\mathbb{F}_q}\{1, x, \dots, x^{h-1}, x^{h}+\eta x^{k+1}, x^{h+1}, \dots, x^{k-1}\}\subset \Span_{\mathbb{F}_q}\{1, x,\dots, x^{k+1}\},$$
    we have $\mathcal{C}_{n,k}(\bm{\alpha}, 2, h, \eta)\subset \ars_{n,k+2}(\bm{\alpha}).$ Therefore, $d=d(\mathcal{C}_{n,k}(\bm{\alpha}, 2, h, \eta))\ge d(\ars_{n,k+2}(\bm{\alpha}))=n-\lceil \frac{k+2}{r}\rceil+1,$ that is,
    \begin{equation}\label{eqn: lower bound on d}
        d\ge n-\left\lceil \frac{k+2}{r}\right\rceil+1.
    \end{equation}
    By the division algorithm, write $k=ar+b$, where $a\in \mathbb{N}\cup \{0\}$ and $0\le b \le r-1$.

    If $1\le b \le r-2,$ then $\left\lceil \frac{k+2}{r}\right\rceil=\left\lceil \frac{ar+b+2}{r}\right\rceil=\left\lceil \frac{ar+b}{r}\right\rceil=\left\lceil \frac{k}{r}\right\rceil.$ It follows from \eqref{eqn: lower bound on d} that $d\ge n-\left\lceil \frac{k}{r}\right\rceil+1$. Since the additive Singleton bound yields $d\le n-\left\lceil \frac{k}{r}\right\rceil+1,$ we conclude that $d=n-\left\lceil \frac{k}{r}\right\rceil+1.$ Therefore, $\mathcal{C}_{n,k}(\bm{\alpha}, 2, h, \eta)$ is additive MDS.

    If $b=0,$ then $\left\lceil \frac{k+2}{r}\right\rceil=\left\lceil \frac{ar+2}{r}\right\rceil=a+1=\left\lceil \frac{k}{r}\right\rceil+1.$ Thus by \eqref{eqn: lower bound on d}, $d\ge n-\left\lceil \frac{k}{r}\right\rceil.$ Combining this with the additive Singleton bound shows that $\mathcal{C}_{n,k}(\bm{\alpha}, 2, h, \eta)$ is either additive MDS or additive almost MDS.

    If $b=r-1,$ then $\left\lceil \frac{k}{r}\right\rceil=\left\lceil \frac{ar+r-1}{r}\right\rceil=a+1,$ and $\left\lceil \frac{k+2}{r}\right\rceil=\left\lceil \frac{ar+r+1}{r}\right\rceil=a+2=\left\lceil \frac{k}{r}\right\rceil+1.$ Therefore, \eqref{eqn: lower bound on d} implies that $d\ge n-\left\lceil \frac{k}{r}\right\rceil.$ Applying the additive Singleton bound once again, we deduce that $\mathcal{C}_{n,k}(\bm{\alpha}, 2, h, \eta)$ is either additive MDS or additive almost MDS. \qed
\end{proof}
\begin{remark}
    By Theorem \ref{thm: AMDS or MDS}, the code $\mathcal{C}_{n,k}(\bm{\alpha}, 2, h, \eta)$ is an additive MDS code over $\mathbb{F}_{q^r}$ for all $1\le k \le nr$ satisfying $k\not\equiv 0,-1 \pmod r.$ 
\end{remark}
\begin{example}
    Let $q=5, r=3, h=1, k=4,$ and let $\eta=1\in \mathbb{F}_5^{*}$. Consider the finite field $\mathbb{F}_{125}:=\mathbb{F}_5[\gamma],$ where $\gamma^3+3\gamma=2.$ Define 
    \begin{equation*}
        \begin{split}
            \bm{\alpha}:=(\gamma, \gamma^{2}, \gamma^{3}, \gamma^{4}, \gamma^{6}, \gamma^{7}, \gamma^{8}, \gamma^{9}, \gamma^{11}, \gamma^{12}, \gamma^{13}, \gamma^{14}, \gamma^{16}, \gamma^{17}, \gamma^{18}, \gamma^{19}, \gamma^{21}, \gamma^{22}, \gamma^{23}, \gamma^{24}, \gamma^{32}, \\
            \gamma^{33}, \gamma^{34}, \gamma^{37}, \gamma^{38}, \gamma^{39}, \gamma^{42}, \gamma^{43}, \gamma^{44}, \gamma^{47}, \gamma^{48}, \gamma^{49}, \gamma^{63}, \gamma^{64}, \gamma^{68}, \gamma^{69}, \gamma^{73}, \gamma^{74}, \gamma^{94}, \gamma^{99}).
        \end{split}
    \end{equation*}
    Using \textsc{Magma} \cite{Magma}, one verifies that the code $\mathcal{C}_{n,k}(\bm{\alpha}, 2, 1, \eta)$ is an $(40, 5^4, 39)$ additive MDS code over $\mathbb{F}_{125}$. This illustrates Theorem \ref{thm: AMDS or MDS}.
\end{example}
The following theorem provides necessary and sufficient conditions for $\mathcal{C}_{n,k}(\bm{\alpha}, 2, h, \eta)$ to be an additive MDS code when $k \equiv 0 \pmod r$, and an additive almost MDS code when $k \equiv -1 \pmod r$.

\begin{theorem}\label{thm: MDS condns for k equiv 0 or -1}
    \begin{enumerate}
        \item[(a)] Suppose that $k\equiv 0\pmod r$. The code $\mathcal{C}_{n,k}(\bm{\alpha}, 2, h, \eta)$ is an additive MDS code over $\mathbb{F}_{q^r}$ if and only if $\eta^{-1}\neq \xi_I,$ for every subset $I\subseteq [n]$ with $|I|=\frac{k}{r},$ where 
        $$
        \xi_I= (-1)^{k-h+1} \left(
\sum_{\substack{
S\subseteq I\times\{0,\ldots,r-1\}\\
|S|=k-h+1
}}
\prod_{(i,j)\in S}\alpha_i^{q^j} -\left({\sum\limits_{i\in I}\Trqr(\alpha_{i})}\right)  
\sum_{\substack{
S\subseteq I\times\{0,\ldots,r-1\}\\
|S|=k-h
}}
\prod_{(i,j)\in S}\alpha_i^{q^j} \right).
        $$
        \item [(b)] Suppose that $k\equiv -1 \pmod r$. The code $\mathcal{C}_{n,k}(\bm{\alpha}, 2, h, \eta)$ is an additive almost MDS code over $\mathbb{F}_{q^r}$ if and only if $\sum\limits_{i\in J}\Trqr(\alpha_{i}) = 0,$ and $\eta^{-1} = \zeta_J$ for some subset $J\subseteq [n]$ with $|J|=\frac{k+1}{r},$ where 
        $$
        \zeta_J= (-1)^{k-h+1}
\sum_{\substack{
S\subseteq J\times\{0,\ldots,r-1\}\\
|S|=k-h+1
}}
\prod_{(i,j)\in S}\alpha_i^{q^j}.
        $$
    \end{enumerate}
\end{theorem}
\begin{proof}
\begin{enumerate}
    \item [(a)] Suppose that $\mathcal{C}_{n,k}(\bm{\alpha}, 2, h, \eta)$ is not additive MDS. Then, by Theorem \ref{thm: AMDS or MDS}, $$d(\mathcal{C}_{n,k}(\bm{\alpha}, 2, h, \eta))=n-\frac{k}{r}.$$ 
    Hence, there exists a nonzero twisted polynomial
    $$
        f(x)=f_0+f_1x+\dots +f_{h-1}x^{h-1}+f_h (x^{h}+\eta x^{k+1})+f_{h+1}x^{h+1}+\cdots+ f_{k-1}x^{k-1}\in \mathcal{V}_{k,2,h,\eta},
    $$
    and a subset $I\subseteq [n]$ with $|I|=\frac{k}{r}$ such that $f(\alpha_i)=0$, for  all $i\in I.$ Since $\alpha_i\in \mathcal{B},$ it follows that $f(\alpha_i^{q^j})=0,$ for every $i\in I$ and $0\le j \le r-1.$ Consequently, $f(x)$ vanishes at $k$ distinct points. Since $\deg f(x)\le k+1,$ we must have $f_h\ne 0;$ otherwise, $\deg(f)\le k-1$, forcing $f(x)$ to be the zero polynomial, a contradiction.

    Define $$g_I(x):=\prod_{i\in I}\prod_{j=0}^{r-1} (x-\alpha_i^{q^j})=g_0+g_1x+ g_2x^2+\cdots+ g_{k-1} x^{k-1}+x^k.$$ Then 
    \begin{equation}\label{eqn: f(x)= blah blah}
        f(x)=\eta f_h (x-a)g_I(x),
    \end{equation}
    for some $a\in \mathbb{F}_q.$ Comparing the coefficient of $x^{k}$ in \eqref{eqn: f(x)= blah blah}, we obtain $0=\eta f_h g_{k-1}-a \eta f_h$ or equivalently, 
    \begin{equation}
        a=g_{k-1}=-\sum_{i\in I}\sum_{j=0}^{r-1} \alpha_i^{q^j}=-\sum_{i\in I} \Trqr (\alpha_i).
    \end{equation}
    Comparing the coefficient of $x^h$ in \eqref{eqn: f(x)= blah blah}, we get
    $f_h=\eta f_h (g_{h-1}-ag_{h}),$ with the convention that $g_{-1}=0.$ Equivalently,
    \begin{align*}
        \eta^{-1}&=g_{h-1}-ag_h\\
        &= (-1)^{k-h+1} 
\sum_{\substack{
S\subseteq I\times\{0,\ldots,r-1\}\\
|S|=k-h+1
}}
\prod_{(i,j)\in S}\alpha_i^{q^j} -a  (-1)^{k-h}
\sum_{\substack{
S\subseteq I\times\{0,\ldots,r-1\}\\
|S|=k-h
}}
\prod_{(i,j)\in S}\alpha_i^{q^j} \\
&= (-1)^{k-h+1} \left(
\sum_{\substack{
S\subseteq I\times\{0,\ldots,r-1\}\\
|S|=k-h+1
}}
\prod_{(i,j)\in S}\alpha_i^{q^j} -\left({\sum\limits_{i\in I}\Trqr(\alpha_{i})}\right)  
\sum_{\substack{
S\subseteq I\times\{0,\ldots,r-1\}\\
|S|=k-h
}}
\prod_{(i,j)\in S}\alpha_i^{q^j} \right),
    \end{align*}

    that is, $\eta^{-1}=\xi_I$, a contradiction. Therefore, $\mathcal{C}_{n,k}(\bm{\alpha}, 2, h, \eta)$ is an additive MDS code over $\mathbb{F}_{q^r}$.

    Conversely, suppose that $\eta^{-1}=\xi_I,$ for some subset $I\subseteq [n]$ with $|I|=\frac{k}{r}$. Then
    \begin{equation}\label{eqn: eta = blah blah}
    \begin{split}
        \eta^{-1}&= (-1)^{k-h+1} \left(
\sum_{\substack{
S\subseteq I\times\{0,\ldots,r-1\}\\
|S|=k-h+1
}}
\prod_{(i,j)\in S}\alpha_i^{q^j} -\left({\sum\limits_{i\in I}\Trqr(\alpha_{i})}\right)  
\sum_{\substack{
S\subseteq I\times\{0,\ldots,r-1\}\\
|S|=k-h
}}
\prod_{(i,j)\in S}\alpha_i^{q^j} \right)\\
&=g_{h-1}+g_h\sum_{i\in I} \Trqr(\alpha_i).
\end{split}
    \end{equation}
    Define
    $$
        f(x)=\eta \left(x+\sum_{i\in I}\Trqr(\alpha_i)\right)\prod_{i\in I}\prod_{j=0}^{r-1} (x-\alpha_i^{q^j})=\eta \left(x+\sum_{i\in I}\Trqr(\alpha_i)\right)g_I(x).
    $$
    Note that the coefficient of $x^k$ in $f(x)$ is $0$, while the coefficient of $x^{k+1}$ is $\eta$. Moreover, by \eqref{eqn: eta = blah blah}, the coefficient of $x^h$ in $f(x)$ is 
    \[
    \eta\left(g_{h-1}+g_{h}\sum\limits_{i\in I}\Trqr(\alpha_i)\right)=\eta \eta^{-1}=1.
    \]
Hence, $f(x)\in \mathcal{V}_{k, 2, h, \eta}$. Since $f(x)$ has $|I|=\frac{k}{r}$ distinct roots $\{\alpha_i: i \in I\},$ it follows that 
\[
d(\mathcal{C}_{n,k}(\bm{\alpha}, 2, h, \eta)) \leq n - \frac{k}{r} = n-\left\lceil\frac{k}{r}\right\rceil.
\]
and consequently the code $\mathcal{C}_{n,k}(\bm{\alpha}, 2, h, \eta)$ is not additive MDS.

\item [(b)]  Suppose that $\mathcal{C}_{n,k}(\bm{\alpha}, 2, h, \eta)$ is an additive almost MDS code. Then there exists a nonzero twisted polynomial
    $$
        f(x)=f_0+f_1x+\dots +f_{h-1}x^{h-1}+f_h (x^{h}+\eta x^{k+1})+f_{h+1}x^{h+1}+\cdots+ f_{k-1}x^{k-1}\in \mathcal{V}_{k, 2, h, \eta},
    $$
    and a subset $J\subseteq [n]$ with $|J|=\frac{k+1}{r}$ such that $f(\alpha_i)=0$, for  all $i\in J.$ Since $\alpha_i\in \mathcal{B},$ it follows that $f(\alpha_i^{q^j})=0,$ for every $i\in J$ and $0\le j \le r-1.$ Consequently, $f(x)$ vanishes at $k+1$ distinct points. Since $\deg f(x)\le k+1,$ we must have $f_h\ne 0.$ 

    Define $$g_J(x):=\prod_{i\in J}\prod_{j=0}^{r-1} (x-\alpha_i^{q^j})=g_0+g_1x+ g_2x^2+\cdots+ g_{k} x^{k}+x^{k+1}.$$ Then 

    \begin{equation}\label{eqn: f(x)_1}
    f(x)=\eta f_h g_J(x).
    \end{equation}
Comparing the coefficient of $x^{k}$ in \eqref{eqn: f(x)_1}, we get $\sum\limits_{i\in J}\Trqr(\alpha_{i}) = 0$, and comparing the coefficient of $x^h$ in \eqref{eqn: f(x)_1} gives $f_{h} = \eta f_{h} g_{h}$, or equivalently,
\begin{equation*}
    \eta^{-1}= g_{h}=(-1)^{k-h+1} 
    \underset{|S|=k-h+1}{\underset{S\subseteq J\times\{0,\ldots,r-1\}}{\sum}}
\prod_{(i,j)\in S}\alpha_i^{q^j},
\end{equation*}
 that is, $\eta^{-1}=\zeta_J$.

  Conversely, suppose that $\sum\limits_{i\in J}\Trqr(\alpha_{i})=0$ and $\eta^{-1}=\zeta_J,$ for some subset $J\subseteq [n]$ with $|J|=\frac{k+1}{r}$. Then
\begin{eqnarray}\label{eqn: eta = blah blah 1}
    \eta^{-1} &=& (-1)^{k-h+1}
\sum_{\substack{
S\subseteq J\times\{0,\ldots,r-1\}\\
|S|=k-h+1
}}
\prod_{(i,j)\in S}\alpha_i^{q^j} \nonumber \\
&=& g_{h}.
\end{eqnarray}
Define
\[
f(x) := \eta \prod_{i\in J}\prod_{j=0}^{r-1} \left(x-\alpha_{i}^{q^j}\right).
\]
Note that the coefficient of $x^{k}$ in $f(x)$ is $0$, while the coefficient of $x^{k+1}$ is $\eta$. Moreover, by  \eqref{eqn: eta = blah blah 1}, the coefficient of $x^h$ in $f(x)$ is 
\[
\eta g_{h} = \eta\eta^{-1} = 1
\]
Hence, $f(x)\in \mathcal{V}_{k, 2, h, \eta}$. Since $f(x)$ has $|J|=\frac{k+1}{r}$ distinct roots $\{\alpha_j: j \in J\},$ it follows that 
\[
d(\mathcal{C}_{n,k}(\bm{\alpha}, 2, h, \eta)) \leq n - \frac{k+1}{r} = n-\left\lceil\frac{k}{r}\right\rceil.
\]
Consequently, $\mathcal{C}_{n,k}(\bm{\alpha}, 2, h, \eta)$ is an additive almost MDS code. \qed 
\end{enumerate}
\end{proof}
\begin{remark}
    Theorems \ref{thm: AMDS or MDS} and \ref{thm: MDS condns for k equiv 0 or -1} generalize \cite[Theorem 3.2]{JiayuMa2026}. Indeed, the latter is recovered as the special case $h=k-1$.
\end{remark}

The following remark guarantees the existence of a choice of $\eta$ satisfying the criterion in Theorem \ref{thm: MDS condns for k equiv 0 or -1}(a).
\begin{remark}\label{remark: Existence}
    Suppose $k\equiv 0 \pmod r$ and let $\mathcal{A}_{n, k, r}:=\left\{\xi_I: I\subseteq [n], \; |I|=\frac{k}{r}\right\}$, where $\xi_I$ is as in Theorem \ref{thm: MDS condns for k equiv 0 or -1}(a).  Then $|\mathcal{A}_{n, k, r}|\le\binom{n}{k/r}$. Hence, if $\binom{n}{k/r}<q-1,$ then $\mathbb{F}_q^{*}\setminus \mathcal{A}_{n, k, r}\ne \emptyset.$ Therefore, there exists $\lambda\in \mathbb{F}_q^{*}\setminus \mathcal{A}_{n, k, r}$. Taking $\eta=\lambda^{-1}$, Theorem \ref{thm: MDS condns for k equiv 0 or -1}(a) implies that $\mathcal{C}_{n,k}(\bm{\alpha}, 2, h, \eta)$ is an additive MDS code over $\mathbb{F}_{q^r}$.
\end{remark}
The existence of additive MDS codes for $k\equiv -1\pmod r$ can be established from Theorem~\ref{thm: MDS condns for k equiv 0 or -1}(b) by an argument analogous to that in Remark \ref{remark: Existence}.
\begin{example}
Let $q=13$, $r=2$, $h=1$, $k=4$, and $n=5$. Since
\[
\binom{n}{k/r}=\binom{5}{2}=10<12=q-1,
\]
the hypothesis of Remark~\ref{remark: Existence} is satisfied. Hence, there exists
$\eta\in\mathbb{F}_{13}^{*}$ such that
$\mathcal{C}_{5,4}(\bm{\alpha},2,1,\eta)$ is an additive $(5,13^4,4)$ MDS code
over $\mathbb{F}_{169}$.
\end{example}

The special case $h=0$ is of particular interest. It admits a simpler characterization and leads to a stronger existence result, allowing additive MDS codes of larger lengths. We therefore state it separately.
\begin{theorem}\label{thm: h=0}
    \begin{enumerate}
        \item[(a)] Suppose that $k\equiv 0\pmod r$. The code $\mathcal{C}_{n,k}(\bm{\alpha}, 2, 0, \eta)$ is an additive MDS code over $\mathbb{F}_{q^r}$ if and only if $\eta^{-1}\neq (-1)^{k} \prod_{i\in I} \Nrqr(\alpha_i) \sum_{i\in I}\Trqr(\alpha_{i}),$ for every subset $I\subseteq [n]$ with $|I|=\frac{k}{r}.$
        \item [(b)] Suppose that $k\equiv -1 \pmod r$. The code $\mathcal{C}_{n,k}(\bm{\alpha}, 2, 0, \eta)$ is an additive almost MDS code over $\mathbb{F}_{q^r}$ if and only if $\sum\limits_{j\in J}\Trqr(\alpha_{j}) = 0$ and $\eta^{-1}= (-1)^{k+1} \prod_{j\in J}\Nrqr(\alpha_j)$ for some subset $J\subseteq [n]$ with $|J|=\frac{k+1}{r}.$
    \end{enumerate}
\end{theorem}
Before giving the existence of additive MDS in Theorem \ref{thm: h=0}(a), we prove the following lemma.
\begin{lemma}\label{lem:betaq-beta}
Let $r$ be such that $\Char(\mathbb{F}_q)\nmid r$. Then, for every
$\beta\in\mathbb{F}_{q^r}$, the element $\beta^q-\beta$ has exactly $r$ distinct conjugates over $\mathbb{F}_q$ if and only if
$\beta\notin\mathbb{F}_{q^d}$ for every proper divisor $d$ of $r$.
\end{lemma}
\begin{proof}
Set $\alpha:=\beta^q-\beta$ throughout the proof.

If $\beta\in\mathbb{F}_{q^d}$ for some proper divisor $d$ of $r$, then $\alpha^{q^d}=\beta^{q^{d+1}}-\beta^{q^d}=\beta^q-\beta=\alpha$, and hence $\alpha$
has less than $r$ distinct conjugates.

Conversely, suppose that $\alpha$ has fewer than $r$ distinct conjugates.
Then there exist integers $0\le j<i\le r-1$ such that $\alpha^{q^i}=\alpha^{q^j}.$ Since $\alpha=\beta^q-\beta$, we obtain $\beta^{q^{i+1}}-\beta^{q^i}
=\beta^{q^{j+1}}-\beta^{q^j},$ or equivalently, $
\beta^{q^{i+1}}-\beta^{q^{j+1}}=\beta^{q^i}-\beta^{q^j}.$ Hence $(\beta^{q^i}-\beta^{q^j})^q
=\beta^{q^i}-\beta^{q^j},$
which implies that $\beta^{q^i}-\beta^{q^j}\in\mathbb{F}_q.$ Raising power $q^{r-j}$ yields
\[
\beta^{q^{i-j}}-\beta\in\mathbb{F}_q.
\]

Set
\[
m:=i-j
\quad\text{and}\quad
\gamma:=\beta^{q^m}-\beta.
\]
Then $1\le m\le r-1$ and let $d=\gcd(m,r)$. Since $m<r$, $d$ is a proper divisor of $r$. Note that $\gamma\in \mathbb{F}_{q}$.
We claim that
\[
\beta^{q^{tm}}=\beta+t\gamma
\]
for every positive integer $t$. The proof is by induction on $t$. The case
$t=1$ is immediate. Assuming the identity holds for some $t\ge1$, we obtain
\[
\beta^{q^{(t+1)m}}
=
(\beta^{q^{tm}})^{q^m}
=
(\beta+t\gamma)^{q^m}
=
\beta^{q^m}+t\gamma
=
\beta+\gamma+t\gamma
=
\beta+(t+1)\gamma,
\]
since $\gamma\in\mathbb{F}_q$.

Taking $t=r$ gives
\[
\beta
=
\beta^{q^{rm}}
=
\beta+r\gamma.
\]
Therefore,
\[
r\gamma=0.
\]
Since $\Char(\mathbb{F}_q)\nmid r$, it follows that $\gamma=0$, that is, $\beta^{q^m}=\beta,$ so $\beta\in \mathbb{F}_{q^m}.$ Since $\beta$ also belongs to $\mathbb{F}_{q^r},$
\[
\beta\in \mathbb{F}_{q^r}\cap \mathbb{F}_{q^m}=\mathbb{F}_{q^{\gcd(m,r)}}
=
\mathbb{F}_{q^d}.
\]

Thus, if $\alpha$ has fewer than $r$ distinct conjugates, then
$\beta\in\mathbb{F}_{q^d}$, for some proper divisor $d$ of $r$.
\flushright\qed
\end{proof}
\begin{lemma}
Let $r$ be such that $\Char(\mathbb{F}_q)\nmid r$. Let
$\mathcal{T}_{q,r}$ be the set of trace-zero elements in $\mathcal{B}_{q,r}$. Then
\[
|\mathcal{T}_{q,r}|=
\frac{1}{rq}
\sum_{d\mid r}
\mu(d)\,q^{r/d}.
\]
\end{lemma}
\begin{proof}
The map
\[
L:\mathbb{F}_{q^r}\longrightarrow\mathbb{F}_{q^r},\qquad
L(\beta)=\beta^q-\beta,
\]
is $\mathbb{F}_q$-linear with kernel $\mathbb{F}_q$. Hence, $|\operatorname{Im}(L)|=q^{r-1}.$ Since $\operatorname{Im}(L)=\{\alpha\in\mathbb{F}_{q^r}:\Tr(\alpha)=0\},$
there are exactly $q^{r-1}$ trace-zero elements in
$\mathbb{F}_{q^r}$. Let
\[
S_r=
\left\{
\beta\in\mathbb{F}_{q^r}:
\beta\notin\mathbb{F}_{q^d}
\text{ for every proper divisor } d \text{ of }r
\right\}.
\]
Then by Lemma~\ref{lem:betaq-beta}, the number of
trace-zero elements having exactly $r$ distinct conjugates over
$\mathbb{F}_q$ is $|L(S_r)|$.

Note that $\beta \in S_r$ implies $\beta+\mathbb{F}_q \subseteq S_r$. Indeed, if
$\beta+c \notin S_r$ for some $c\in\mathbb{F}_q$, then
$\beta+c\in\mathbb{F}_{q^d}$ for some proper divisor $d$ of $r$. Since
$\mathbb{F}_q\subseteq\mathbb{F}_{q^d}$, it follows that
\[
\beta=(\beta+c)-c\in\mathbb{F}_{q^d},
\]
contradicting the assumption that $\beta\in S_r$. Hence,
$\beta+\mathbb{F}_q\subseteq S_r$. Therefore, $S_{r}$ is a disjoint union of complete cosets of $\mathbb{F}_{q}$, each of size $q$. Hence, $|L(S_r)| = \frac{|S_r|}{q}$.

Now, for every divisor $d$ of $r,$ define
\[
S_d:=
\left\{
\beta\in\mathbb{F}_{q^d}:
\beta\notin\mathbb{F}_{q^e}
\text{ for every proper divisor } e \text{ of }d
\right\}.
\]
Every element of $\mathbb F_{q^r}$ has a unique minimal subfield
$\mathbb F_{q^d}$, where $d\mid r$.
Therefore the sets $S_d$, $d\mid r$, are pairwise disjoint and
\[
\mathbb F_{q^r}
=
\bigsqcup_{d\mid r}S_d.
\]
Consequently,
\[
q^r=\sum_{d|r} |S_d|.
\]
By the M\"obius inversion formula, we obtain
\[
|S_r|
=
\sum_{d\mid r}\mu(d)\,q^{r/d},
\]
and consequently, the number of trace-zero elements having exactly
$r$ distinct conjugates is
\[
\frac{1}{q}
\sum_{d\mid r}
\mu(d)\,q^{r/d}.
\]
Therefore, these elements split
into disjoint Frobenius orbits, each of size $r$. Choosing one representative
from each orbit yields a subset $\mathcal{T}_{q,r}\subseteq\mathcal{B}_{q,r}$
consisting of trace-zero elements, and
\[
|\mathcal{T}_{q,r}|=
\frac{1}{rq}
\sum_{d\mid r}
\mu(d)\,q^{r/d}.
\]
This completes the proof.\qed
\end{proof}
We now obtain a family of additive MDS codes of length $n=\frac{1}{rq}\sum\limits_{d\mid r}\mu(d)\,q^{r/d}$ over $\mathbb{F}_{q^r}$ in the case $h=0$, where $r$ is a positive integer such that
$\Char(\mathbb{F}_q)\nmid r$. 
\begin{corollary}\label{corr: Existence for h=0}
Let $r$ be such that $\Char(\mathbb{F}_q)\nmid r$ and suppose that
$k\equiv0\pmod r$. Then the code $\mathcal{C}_{|\mathcal{T}_{q,r}|,k}
(\mathcal{T}_{q,r},2,0,\eta)$ is an additive MDS code over $\mathbb{F}_{q^r}$ for every
$\eta\in\mathbb{F}_q^*$.
\end{corollary}

\begin{proof}
Set $\bm{\alpha}:=\mathcal{T}_{q,r}$ and
$n:=|\mathcal{T}_{q,r}|$. Since
$\Trqr(\alpha_i)=0$ for every $\alpha_i\in\mathcal{T}_{q,r}$, it follows that
\[
(-1)^k
\prod_{i\in I}\Nrqr(\alpha_i)
\sum_{i\in I}\Trqr(\alpha_i)
=0
\]
for every subset $I\subseteq[n]$ with
$|I|=\frac{k}{r}$.
Hence,
\[
\eta^{-1}
\neq
(-1)^k
\prod_{i\in I}\Nrqr(\alpha_i)
\sum_{i\in I}\Trqr(\alpha_i)
\]
for every $\eta\in\mathbb{F}_q^*$.
The conclusion now follows immediately from
Theorem~\ref{thm: h=0}(a).\qed
\end{proof}
\begin{example}
    Let $q=3, r=5, h=0, k=5,$ and let $\eta=2\in \mathbb{F}_3$. Consider the finite field $\mathbb{F}_{243}:=\mathbb{F}_3[\gamma],$ where $\gamma^5+2\gamma=2.$ Then 
    $$\mathcal{T}_{3, 5}:= \{\gamma^{75}, \gamma^{122}, \gamma^{71}, \gamma^{174}, \gamma^{47}, \gamma^{33}, \gamma^{210}, \gamma^{3}, \gamma^{175}, \gamma^{83}, \gamma^{20}, \gamma^2,\gamma^{179}, \gamma^{68}, \gamma^{192}, \gamma^{44}\}.$$
    Using \textsc{Magma} \cite{Magma}, one verifies that the code $\mathcal{C}_{16,5}(\mathcal{T}_{3,5}, 2, 0, \eta)$ is an $(16, 3^5, 16)$ additive MDS code over $\mathbb{F}_{243}$, in accordance with Corollary~\ref{corr: Existence for h=0}.
\end{example}
\begin{example}
    Let $q=5, r=4, h=0, k=4,$ and let $\eta=3\in \mathbb{F}_5$. Consider the finite field $\mathbb{F}_{625}:=\mathbb{F}_5[\gamma],$ where $\gamma^4+4\gamma^2+4\gamma=3.$ Then
    \begin{equation*}
        \begin{split}
            \mathcal{T}_{5, 4}:=\{\gamma^{281}, \gamma^{164}, \gamma^{466}, \gamma^{579}, \gamma^{512}, \gamma^{44}, \gamma^{221}, \gamma^{283}, \gamma^{403}, \gamma^{591}, \gamma^{262},
            \gamma^{376}, \gamma^{106}, \gamma^{587}, \gamma^{464},\\
            \gamma^{523}, \gamma^{147}, \gamma^{325}, \gamma^{39}, \gamma^{559}, \gamma^{152}, \gamma^{585},
            \gamma^{313}, \gamma^{435}, \gamma^{604}, \gamma^{11}, \gamma^{614}, \gamma^{493}, \gamma^{125}, \gamma^{436}\}.
        \end{split}
    \end{equation*}
    Using \textsc{Magma} \cite{Magma}, one verifies that the code $\mathcal{C}_{30, 4}(\mathcal{T}_{5, 4}, 2, 0, \eta)$ is an $(30, 5^4, 30)$ additive MDS code over $\mathbb{F}_{625}$, in accordance with Corollary~\ref{corr: Existence for h=0}.
\end{example}
We use Theorem \ref{thm: h=0}(b) to construct additive MDS $\mathcal{C}_{n,k}(\bm{\alpha}, 2, 0, \eta)$ codes when $k\equiv -1 \pmod r.$ We begin with the following lemma.

\begin{lemma}\label{lem: Norm 1}
Let
$\mathcal{S}_{q,r}$ be the set of norm-one elements in $\mathcal{B}_{q,r}$. Then
\[
|\mathcal{S}_{q,r}|=
\frac{1}{r(q-1)}
\sum_{d\mid r}\mu(d)\gcd(d, q-1)(q^{r/d}-1).
\]
\end{lemma}
\begin{proof}
Let $H=\{\alpha\in\mathbb{F}_{q^r}: \Nrqr(\alpha)=1\}$, and for every divisor $d$ of $r,$ define
\[
S_d:=
\left\{
\beta\in\mathbb{F}_{q^d}:
\beta\notin\mathbb{F}_{q^e}
\text{ for every proper divisor } e \text{ of }d
\right\}.
\]
Note that the number of elements of $\mathbb{F}_{q^r}$ having $r$ distinct conjugates and norm-one equals $|H\cap S_{r}|$.

Every element of $\mathbb F_{q^r}$ has a unique minimal subfield
$\mathbb F_{q^d}$, where $d\mid r$.
Therefore the sets $S_d$, $d\mid r$, are pairwise disjoint and
\[
\mathbb F_{q^r}
=
\bigsqcup_{d\mid r}S_d.
\]
Consequently,
\[
H \cap \mathbb{F}_{q^r}=\bigsqcup_{d\mid r}(H \cap S_d).
\]
Hence, 
\[
|H \cap \mathbb{F}_{q^r}|=\sum_{d\mid r} |H\cap S_d|.
\]
By the M\"obius inversion formula, we obtain
\begin{equation}\label{eqn: H cap Sr}
    |H \cap S_r|
=
\sum_{d\mid r}\mu(d)\,|H \cap \mathbb{F}_{q^{r/d}}|.
\end{equation}
Now, for every divisor $d$ of $r$, let $N_d =  |H\,\cap\,\mathbb{F}_{q^d}|.$ Since for every \(\alpha\in\mathbb{F}_{q^d}\), $\N_{q^r/q^d}(\alpha)
=\alpha^{r/d}$, and $\Nrqr
=
\Nrqm
\circ
\N_{q^r/q^d},$
it follows that
\[
\Nrqr(\alpha)
=
\Nrqm(\alpha)^{\,r/d}.
\]
Hence,
\[
H\cap\mathbb{F}_{q^d}
=
\left\{
\alpha\in\mathbb{F}_{q^d}:
\Nrqm(\alpha)^{\,r/d}=1
\right\}.
\]

The norm map
\[
\Nrqm:
\mathbb{F}_{q^d}^{*}\longrightarrow\mathbb{F}_{q}^{*}
\]
is surjective, and each fibre has cardinality $\frac{q^d-1}{q-1}.$ Since \(\mathbb{F}_{q}^{*}\) is a cyclic group of order \(q-1\), the equation $x^{\,r/d}=1$ has exactly $\gcd\!\left(\frac{r}{d},q-1\right)$ solutions in \(\mathbb{F}_{q}^{*}\). Therefore,
\begin{equation}\label{eqn: nd}
    N_d
=
\frac{q^d-1}{q-1}
\gcd\!\left(\frac{r}{d},q-1\right).
\end{equation}
Using \eqref{eqn: nd} in \eqref{eqn: H cap Sr}, we get
\[
|H \cap S_r|
=
\sum_{d\mid r}\mu(d)\,N_{r/d} = \frac{1}{q-1} \sum_{d\mid r}\mu(d)\,\gcd(d, q-1) (q^{r/d}-1).
\]
Since every orbit has cardinality $r$, we obtain
\[
|\mathcal{S}_{q,r}|
=
\frac{|H\cap S_{r}|}{r}
=
\frac{1}{r(q-1)}
\sum_{d\mid r}\mu(d)\gcd (d,q-1)(q^{r/d}-1).
\]
This completes the proof. \qed
\end{proof}
We now obtain a family of additive MDS codes of length $n=\frac{1}{r(q-1)}
\sum\limits_{d\mid r}\mu(d)\gcd (d,q-1)\left(q^{r/d}-1\right)$ over $\mathbb{F}_{q^r}$ in the case $h=0$, when $q>3.$
\begin{corollary}\label{corr: Existence h=0 family 2}
Let $q>3$ and suppose that
$k\equiv -1\pmod r$. Then the code $\mathcal{C}_{|\mathcal{S}_{q,r}|,k}
(\mathcal{S}_{q,r},2,0,\eta)$ is an additive MDS code over $\mathbb{F}_{q^r}$ for every
$\eta\in\mathbb{F}_q^*\setminus\{\pm1\}$.
\end{corollary}
\begin{proof}
Set $\bm{\alpha}:=\mathcal{S}_{q,r}$ and
$n:=|\mathcal{S}_{q,r}|$. Since
$\Nrqr(\alpha_i)=1$ for every $\alpha_i\in\mathcal{S}_{q,r}$, it follows that
\[
(-1)^{k+1} \prod_{j\in J}\Nrqr(\alpha_j)=\pm1,
\]
for every subset $J\subseteq[n]$ with
$|J|=\frac{k+1}{r}$.
Since $\eta\in\mathbb{F}_q^{*}\setminus\{\pm1\},$
\[
\eta^{-1}
\neq
(-1)^{k+1} \prod_{j\in J}\Nrqr(\alpha_j),
\]
for every subset $J\subseteq[n]$ with
$|J|=\frac{k+1}{r}$.
Consequently, the second condition in
Theorem~\ref{thm: h=0}(b) (that is, $\eta^{-1}=
(-1)^{k+1} \prod_{j\in J}\Nrqr(\alpha_j)$, for some subset $J\subseteq[n]$ with
$|J|=\frac{k+1}{r}$) fails. Thus,
$\mathcal{C}_{n,k}(\bm{\alpha},2,0,\eta)$ is not an additive almost MDS code.
Since every code in this family is either an additive MDS or an additive almost MDS by Theorem~\ref{thm: AMDS or MDS}, the code is necessarily an additive MDS code. \qed
\end{proof}
\begin{example}
    Let $q=4, r=4, h=0, k=3,$ and let $\eta=\beta\in \mathbb{F}_4=\mathbb{F}_2[\beta]$, where $\beta^2+\beta=1$. Consider the finite field $\mathbb{F}_{256}:=\mathbb{F}_4[\gamma],$ where $\gamma^4+\gamma^3+\beta(\gamma^2+\gamma+1)=0.$ Then
    \begin{equation*}
        \begin{split}
            \mathcal{S}_{4, 4}:=\{\gamma^{129}, \gamma^{3}, \gamma^{132}, \gamma^{135}, \gamma^{9}, \gamma^{138}, \gamma^{141}, \gamma^{15}, \gamma^{147}, \gamma^{21}, \gamma^{150}, \gamma^{27},\\
            \gamma^{156}, \gamma^{159}, \gamma^{171}, \gamma^{45},\gamma^{183}, \gamma^{189}, \gamma^{63}, \gamma^{213}\}.
        \end{split}
    \end{equation*}
    Using \textsc{Magma}, one verifies that the code $\mathcal{C}_{20,3}(\mathcal{S}_{4, 4}, 2, 0, \eta)$ is an $(20, 4^3, 20)$ additive MDS code over $\mathbb{F}_{256}$, in accordance with Corollary~\ref{corr: Existence h=0 family 2}.
\end{example}
\begin{theorem}
    With the notations as above, if $3<k<\frac{nr}{2},$ then the code $\mathcal{C}_{n,k}(\bm{\alpha}, 2, h, \eta)$ is not monomially equivalent to an additive RS code.
\end{theorem}
\begin{proof}
    We prove the claim by showing that $\dim_{\mathbb{F}_q}\! \mathcal{C}_{n,k}(\bm{\alpha}, 2, h, \eta)^{\star 2}\ge 2k.$ The desired conclusion then follows from Lemma \ref{lem: equiv implies dim of schur square equal} and Lemma \ref{lem: dim of schur of ars}.
    
    For each $i\ge 0,$ let $\bm{\alpha}^i:=(\alpha_1^i, \alpha_2^i, \dots, \alpha_n^i)$. A generator matrix of $\mathcal{C}_{n,k}(\bm{\alpha}, 2, h, \eta)$ is
    $$
        G=\begin{pmatrix}
            1 & \cdots & 1 \\
            \alpha_{1} & \cdots & \alpha_{n} \\
            \vdots & \cdots & \vdots \\
            \alpha_{1}^{h-1} & \cdots & \alpha_{n}^{h-1} \\
            \alpha_{1}^{h}+\eta \alpha_{1}^{k+1} & \cdots & \alpha_{n}^{h}+\eta \alpha_{n}^{k+1} \\
            \alpha_{1}^{h+1} & \cdots & \alpha_{n}^{h+1} \\
            \vdots & \cdots & \vdots \\
            \alpha_{1}^{k-1} & \cdots & \alpha_{n}^{k-1}
        \end{pmatrix}
        \quad = \quad \left(
\begin{tabular}{c}
     $\bm{g}_{1}$ \\
     $\bm{g}_{2}$ \\
     $\vdots$ \\
      $\bm{g}_{h}$ \\
       $\bm{g}_{h+1}$ \\
        $\bm{g}_{h+2}$ \\
        $\vdots$ \\
     $\bm{g}_{k}$ 
\end{tabular}
\right),
    $$
    where $\bm{g}_i=\bm{\alpha}^{i-1},$ $1\le i\le k,$ $i\ne h+1,$ and $\bm{g}_{h+1}:=\bm{\alpha}^h+\eta \bm{\alpha}^{k+1}.$
    
    For $h\notin\{0, k-1\},$ we claim that 
    \begin{equation*}
        \mathcal{D}_{h}:=\{\bm{g}_1\star \bm{g}_1,\ \bm{g}_1\star \bm{g}_2, \dots, \bm{g}_1\star \bm{g}_{k}\} \cup \{ \bm{g}_2\star \bm{g}_{k}, \dots, \bm{g}_{k}\star \bm{g}_{k}\}\cup \{ \bm{g}_{h+1}\star \bm{g}_{k-1}\} \subset \mathcal{C}_{n,k}(\bm{\alpha}, 2, h, \eta)^{\star 2}
    \end{equation*}
    is linearly independent over $\mathbb{F}_q$. Observe that
    \[
    \begin{aligned}
\mathcal{D}_h
={}&
\{\bm{1},\bm{\alpha},\ldots,\bm{\alpha}^{h-1},
\bm{\alpha}^{h}+\eta\bm{\alpha}^{k+1},
\bm{\alpha}^{h+1},\ldots,\bm{\alpha}^{k-1}\}
\\
&\cup
\{\bm{\alpha}^{k},\ldots,\bm{\alpha}^{h+k-2},
\bm{\alpha}^{h+k-1}+\eta\bm{\alpha}^{2k},
\bm{\alpha}^{h+k},\ldots,\bm{\alpha}^{2k-2}\}
\\
&\cup
\{\bm{\alpha}^{h+k-2}+\eta\bm{\alpha}^{2k-1}\}.
\end{aligned}
    \]
    Suppose that $\mathcal{D}_h$ is linearly dependent over $\mathbb{F}_q.$ Then there exist coefficients $a_0,a_1,\ldots,a_{2k-1}\in\mathbb{F}_q$, not all zero, such that
    \[
        \begin{aligned}
        &a_0\bm{1}
        +a_1\bm{\alpha}
        +\cdots
        +a_{h-1}\bm{\alpha}^{h-1}
        +a_h(\bm{\alpha}^{h}+\eta\bm{\alpha}^{k+1})
        +a_{h+1}\bm{\alpha}^{h+1}
        +\cdots
        +a_{h+k-2}\bm{\alpha}^{h+k-2}
        \\
        &\qquad
        +a_{h+k-1}(\bm{\alpha}^{h+k-1}+\eta\bm{\alpha}^{2k})
        +a_{h+k}\bm{\alpha}^{h+k}
        +\cdots
        +a_{2k-2}\bm{\alpha}^{2k-2}
        +a_{2k-1}(\bm{\alpha}^{h+k-2}+\eta\bm{\alpha}^{2k-1})
        =\bm{0}.
        \end{aligned}
\]
Equivalently, for every $1\le i\le n$, $f(\alpha_i)=0,$ where
\[
\begin{aligned}
f(x)
={}&
a_0+a_1x+\cdots+a_{h-1}x^{h-1}
+a_h(x^h+\eta x^{k+1})
+a_{h+1}x^{h+1}
+\cdots
+a_{h+k-2}x^{h+k-2}
\\
&
+a_{h+k-1}(x^{h+k-1}+\eta x^{2k})
+a_{h+k}x^{h+k}
+\cdots
+a_{2k-2}x^{2k-2}
+a_{2k-1}(x^{h+k-2}+\eta x^{2k-1}).
\end{aligned}
\]
Since each $\alpha_i\in\mathcal{B}$ and the coefficients of $f(x)$ lie in
$\mathbb{F}_q$, it follows that
\[
f(\alpha_i^{q^j})=f(\alpha_i)^{q^j}=0,
\qquad
1\le i\le n,\;0\le j\le r-1.
\]
Hence $f(x)$ has at least $nr$ distinct roots. On the other hand, $\deg(f)\le 2k<nr.$
Therefore, $f(x)\equiv0$, forcing all coefficients
$a_0,a_1,\ldots,a_{2k-1}$ to be zero, a contradiction. Thus,
$\mathcal{D}_h$ is linearly independent over $\mathbb{F}_q$.

Since $|\mathcal{D}_h|=2k$, we conclude that $
\dim_{\mathbb{F}_q}
\bigl(\mathcal{C}_{n,k}(\bm{\alpha},2,h,\eta)^{\star2}\bigr)
\ge 2k.$

Similarly, for $h=0$, one can show that
    $$
    \mathcal{D}_0:=\{\bm{g}_2\star \bm{g}_1, \ \bm{g}_2\star \bm{g}_2, \dots, \bm{g}_2\star \bm{g}_k\} \cup \{ \bm{g}_3\star \bm{g}_k, \dots, \bm{g}_k\star \bm{g}_k\}\cup \{ \bm{g}_1\star \bm{g}_{k-1},\ \bm{g}_1\star \bm{g}_k\} \subset \mathcal{C}_{n,k}(\bm{\alpha}, 2, 0, \eta)^{\star 2}
    $$
    is linearly independent over $\mathbb{F}_q$.

Finally, assume that $h=k-1$. One can again prove that
\[
    \begin{aligned}
        \mathcal{D}_{k-1}:={}&
        \{\bm{g}_1\star\bm{g}_1,\bm{g}_1\star\bm{g}_2,\ldots,\bm{g}_1\star\bm{g}_{k-1}\}
        \cup
        \{\bm{g}_2\star\bm{g}_{k-1},\ldots,\bm{g}_{k-1}\star\bm{g}_{k-1}\}
        \\
        &
        \cup
        \{\bm{g}_k\star\bm{g}_{k-3},
        \bm{g}_k\star\bm{g}_{k-2},
        \bm{g}_k\star\bm{g}_{k-1}\}
        \subset
        \mathcal{C}_{n,k}(\bm{\alpha},2,k-1,\eta)^{\star2}
    \end{aligned}
\]
    is linearly independent over $\mathbb{F}_q$. \qed
\end{proof}
\section{A Class of Additive MDS Double-Twisted RS Codes}\label{Section 4}
In this section, we consider the case $\ell=2$, $\bm{t}=(1,2)$, $\bm{h}=(0,0)$, and $\bm{\eta}=(\eta_1,\eta_2)\in(\mathbb{F}_q^*)^2$. To simplify the notation, we write
\[
\mathcal{C}_{n,k}(\bm{\alpha},\bm{t},\bm{h},\bm{\eta})
:=
\atrs_{n,k}(\bm{\alpha},\bm{t},\bm{h},\bm{\eta}),
\]
where $\bm{t}=(1,2)$ and $\bm{h}=(0,0)$ throughout this section. We construct additive MDS codes in the family $\mathcal{C}_{n,k}(\bm{\alpha},\bm{t},\bm{h},\bm{\eta})$ and show that they are not monomially equivalent to additive Reed--Solomon codes.
\begin{theorem}\label{thm:AMDS or MDS_ double twist}
     For $1\le k\le nr,$ the code $\mathcal{C}_{n,k}(\bm{\alpha}, \bm{t}, \bm{h}, \bm{\eta})$ over $\mathbb{F}_{q^r}$ is either an additive MDS code or an additive almost MDS code.
\end{theorem}
\begin{proof}
 Since 
    $$\Span_{\mathbb{F}_q}\{1+\eta_{1}x^{k}+\eta_{2}x^{k+1}, x, \dots, x^{k-1}\}\subset \Span_{\mathbb{F}_q}\{1, x,\dots, x^{k+1}\},$$
    we have $\mathcal{C}_{n,k}(\bm{\alpha}, \bm{t}, \bm{h}, \bm{\eta})\subset \ars_{n,k+2}(\bm{\alpha}).$ Therefore, $d=d(\mathcal{C}_{n,k}(\bm{\alpha}, \bm{t}, \bm{h}, \bm{\eta}))\ge d(\ars_{n,k+2}(\bm{\alpha}))=n-\lceil \frac{k+2}{r}\rceil+1,$ that is,
    \begin{equation*}
        d\ge n-\left\lceil \frac{k+2}{r}\right\rceil+1.
    \end{equation*}
    The proof now follows verbatim from Theorem \ref{thm: AMDS or MDS}. \qed
\end{proof}
\begin{remark}
         The code $\mathcal{C}_{n,k}(\bm{\alpha}, \bm{t}, \bm{h}, \bm{\eta})$ is an additive MDS code over $\mathbb{F}_{q^r}$ for all $1\le k \le nr$ satisfying $k\not\equiv 0,-1 \pmod r.$ 
\end{remark}

The following theorem gives necessary and sufficient conditions for the code $\mathcal{C}_{n,k}(\bm{\alpha},\bm{t},\bm{h},\bm{\eta})$ to be additive MDS when $k\equiv0\pmod r$, and additive almost MDS when $k\equiv-1\pmod r$.
\begin{theorem}\label{thm: double twist}
    \begin{enumerate}
        \item[(a)] Suppose that $k\equiv 0\pmod r$. The code $\mathcal{C}_{n,k}(\bm{\alpha}, \bm{t}, \bm{h}, \bm{\eta})$ is an additive MDS code over $\mathbb{F}_{q^r}$ if and only if 
        $$
         (-1)^{k} \left(\eta_{2}{\sum\limits_{i\in I}\Trqr(\alpha_{i})} + \eta_{1}\right) 
\prod_{i\in I}\Nrqr(\alpha_{i}) \neq 1,
        $$
for every subset $I\subseteq [n]$ with $|I|=\frac{k}{r}.$
        \item[(b)]  Suppose that $k\equiv -1\pmod r$. The code $\mathcal{C}_{n,k}(\bm{\alpha}, \bm{t}, \bm{h}, \bm{\eta})$ is an additive almost MDS code over $\mathbb{F}_{q^r}$ if and only if  
        \[
\eta_{1}+\eta_{2}\sum_{i\in J}\Trqr(\alpha_i)=0
\quad\text{and}\quad
\eta_{2}^{-1}=(-1)^{k+1}\prod_{i\in J}\Nrqr(\alpha_i).
\]
for some subset $J\subseteq [n]$ with $|J|=\frac{k+1}{r}.$ 
    \end{enumerate}
\end{theorem}
\begin{proof}
    \begin{enumerate}
        \item[(a)] Suppose that $\mathcal{C}_{n,k}(\bm{\alpha}, \bm{t}, \bm{h}, \bm{\eta})$ is not additive MDS. Then by Theorem \ref{thm:AMDS or MDS_ double twist},
        \[
        d(\mathcal{C}_{n,k}(\bm{\alpha}, \bm{t}, \bm{h}, \bm{\eta})) = n- \frac{k}{r}.
        \]
        Thus, there exists a nonzero twisted polynomial
        \[
         f(x)=f_{0}(1+\eta_{1}x^{k}+\eta_{2}x^{k+1})+f_1x+\cdots+ f_{k-1}x^{k-1} \in \twist,
        \]
 and a subset $I\subseteq [n]$ with $|I|=\frac{k}{r}$ such that $f(\alpha_i)=0$, for  all $i\in I.$ Since $\alpha_i\in \mathcal{B},$ it follows that $f(\alpha_i^{q^j})=0,$ for every $i\in I$ and $0\le j \le r-1.$ Consequently, $f(x)$ vanishes at $k$ distinct points. Since $\deg f(x)\le k+1,$ we must have $f_0\ne 0$.
 Define $$g_I(x):=\prod_{i\in I}\prod_{j=0}^{r-1} (x-\alpha_i^{q^j})=g_0+g_1x+ g_2x^2+\cdots+ g_{k-1} x^{k-1}+x^k.$$ Then
 \begin{equation}\label{eqn: f(x)= blah blah _ double twist}
        f(x)=\eta_{2} f_0 (x-a)g_I(x),
    \end{equation}
for some $a\in\mathbb{F}_{q}$. Comparing the coefficient of $x^{k}$ in \eqref{eqn: f(x)= blah blah _ double twist}, we obtain
\[
\eta_{1} = \eta_{2}(g_{k-1}-a),
\]
or equivalently,
\[
a = -\left(\sum\limits_{i\in I}\Trqr(\alpha_{i})+\eta_{1}\eta_{2}^{-1}\right).
\]
Next, comparing the constant term in \eqref{eqn: f(x)= blah blah _ double twist}, we get 
\[
f_{0} = -\eta_{2}f_{0}ag_{0},
\]
or equivalently,
\[
1 = -\eta_{2}a g_{0}.
\] 
Substituting the values of $a$ and $g_{0}$ into the preceding equation, we obtain
\begin{eqnarray*}
    1 &=& \eta_{2}\left(\sum\limits_{i\in I}\Trqr(\alpha_{i})+\eta_{1}\eta_{2}^{-1}\right) (-1)^{k}\prod_{i\in I}\Nrqr(\alpha_{i})\\
    &=& (-1)^{k}\left(\eta_{2}\sum\limits_{i\in I}\Trqr(\alpha_{i})+\eta_{1}\right) \prod_{i\in I}\Nrqr(\alpha_{i}),
\end{eqnarray*}
which contradicts the hypothesis. Therefore, $\mathcal{C}_{n,k}(\bm{\alpha}, \bm{t}, \bm{h}, \bm{\eta})$ is an additive MDS code.

Conversely, suppose that
\begin{eqnarray*}
    1 &=& (-1)^{k}\left(\eta_{2}\sum\limits_{i\in I}\Trqr(\alpha_{i})+\eta_{1}\right) \prod_{i\in I}\Nrqr(\alpha_{i}) \\
    &=& \eta_{2}\left(\sum\limits_{i\in I}\Trqr(\alpha_{i})+\eta_{1}\eta_{2}^{-1}\right)g_{0}
\end{eqnarray*}
for some subset $I\subseteq [n]$ with $|I|=\frac{k}{r}$. Define 
\begin{eqnarray*}
     f(x)&:=&\eta_{2} \left(x+\sum_{i\in I}\Trqr(\alpha_i)+\eta_{1}\eta_{2}^{-1}\right)\prod_{i\in I}\prod_{j=0}^{r-1} (x-\alpha_i^{q^j})\\
      &=&\eta_{2} \left(x+\sum_{i\in I}\Trqr(\alpha_i)+\eta_{1}\eta_{2}^{-1}\right)g_I(x).
\end{eqnarray*}
Note that the coefficient of $x^{k+1}$ in $f(x)$ is $\eta_{2}$, while the coefficient of $x^{k}$ is $\eta_{1}$. Moreover, the constant coefficient is 
\[
\eta_{2}\left(\sum\limits_{i\in I}\Trqr(\alpha_i)+\eta_{1}\eta_{2}^{-1}\right)g_{0} = 1
\]
Hence, $f(x)\in \twist$. Since $f(x)$ has $|I|=\frac{k}{r}$ distinct roots $\{\alpha_i: i \in I\},$ it follows that 
\[
d(\mathcal{C}_{n,k}(\bm{\alpha}, \bm{t}, \bm{h}, \bm{\eta})) \leq n - \frac{k}{r} = n-\left\lceil\frac{k}{r}\right\rceil.
\]
and consequently the code $\mathcal{C}_{n,k}(\bm{\alpha}, \bm{t}, \bm{h}, \bm{\eta})$ is not additive MDS. 
\item[(b)] Suppose that $\mathcal{C}_{n,k}(\bm{\alpha}, \bm{t}, \bm{h}, \bm{\eta})$ is additive almost MDS. Then there exists a nonzero twisted polynomial
 $$             f(x)=f_0(1+\eta_{1}x^{k}+\eta_{2}x^{k+1})+f_1x+\cdots+ f_{k-1}x^{k-1} \in \twist,
$$
     and a subset $J\subseteq [n]$ with $|J|=\frac{k+1}{r}$ such that $f(\alpha_i)=0$, for  all $i\in J.$ Since $\alpha_i\in \mathcal{B},$ it follows that $f(\alpha_i^{q^j})=0,$ for every $i\in J$ and $0\le j \le r-1.$ Consequently, $f(x)$ vanishes at $k+1$ distinct points. Since $\deg f(x)\le k+1,$ we must have $f_0\ne 0.$ 

     Define $$g_J(x):=\prod_{i\in J}\prod_{j=0}^{r-1} (x-\alpha_i^{q^j})=g_0+g_1x+ g_2x^2+\cdots+ g_{k} x^{k}+x^{k+1}.$$ Then 
     \begin{equation}\label{eqn: f(x)_1_double twist}
    f(x)=\eta_{2} f_{0} g_J(x).
    \end{equation}

    Comparing the coefficient of $x^{k}$ in \eqref{eqn: f(x)_1_double twist}, we get
    \[
    \eta_{1}f_{0} = \eta_{2}f_{0}g_{k-1},
    \]
    or equivalently,
    \[
    \eta_{1}+\eta_{2}\sum\limits_{i\in J}\Trqr(\alpha_{i}) = 0.
    \]
    Next, comparing the constant term in \eqref{eqn: f(x)_1_double twist}, we get
    \[
    f_{0} = \eta_{2}f_{0}g_{0},
    \]
    or equivalently, 
    \[
    \eta_{2}^{-1} = (-1)^{k+1}\prod_{i\in J}\Nrqr(\alpha_{i}).
    \]
    Conversely, suppose that 
    \begin{equation} \label{eqn: eta_1 + eta_2}
  \eta_{1}+\eta_{2}\sum_{i\in J}\Trqr(\alpha_i)=0
\quad\text{and}\quad
\eta_{2}^{-1}=(-1)^{k+1}\prod_{i\in J}\Nrqr(\alpha_i),
    \end{equation}

for some subset $J\subseteq[n]$ with  $|J|=\frac{k+1}{r}$. Define
\[
f(x) := \eta_{2} \prod_{i\in J}\prod_{j=0}^{r-1} \left(x-\alpha_{i}^{q^j}\right).
\]
Note that the coefficient of $x^{k+1}$ in $f(x)$ is $\eta_{2}$. By \eqref{eqn: eta_1 + eta_2}, the coefficient of $x^{k}$ is $\eta_{1}$. Moreover, by \eqref{eqn: eta_1 + eta_2}, the constant coefficient is $\eta_{2}g_{0}=1.$
Hence, $f(x)\in \twist$. Since $f(x)$ has $|J|=\frac{k+1}{r}$ distinct roots $\{\alpha_j: j \in J\},$ it follows that 
\[
d(\mathcal{C}_{n,k}(\bm{\alpha}, \bm{t}, \bm{h}, \bm{\eta})) \leq n - \frac{k+1}{r} = n-\left\lceil\frac{k}{r}\right\rceil.
\]
Consequently, the code $\mathcal{C}_{n,k}(\bm{\alpha}, \bm{t}, \bm{h}, \bm{\eta})$ is additive almost MDS. \qed
    \end{enumerate}
\end{proof}
In the following propositions, we establish the existence of additive MDS codes for the case $k\equiv -1 \pmod r$ using Theorem \ref{thm: double twist}(b).
\begin{proposition}\label{prop: existence1}
    Let $r$ be such that $\Char(\mathbb{F}_q)\nmid r$ and suppose that
$k\equiv -1 \pmod r$. Then the code $\mathcal{C}_{|\mathcal{T}_{q,r}|,k}
(\mathcal{T}_{q,r}, \bm{t}, \bm{h}, \bm{\eta})$ is an additive MDS code over $\mathbb{F}_{q^r}$ for every
$\bm{\eta} = (\eta_{1},\eta_{2})\in(\mathbb{F}_q^*)^2$.
\end{proposition}
\begin{proof}
    Set $\bm{\alpha}:=\mathcal{T}_{q,r}$ and
$n:=|\mathcal{T}_{q,r}|$. Since
$\Trqr(\alpha_i)=0$ for every $\alpha_i\in\mathcal{T}_{q,r}$, it follows that
\[
 \eta_{2}
\sum_{i\in I}\Trqr(\alpha_i)
=0
\]
for every subset $I\subseteq[n]$ with
$|I|=\frac{k}{r}$.
Hence,
\[
\eta_{1} + \eta_{2}
\sum_{i\in I}\Trqr(\alpha_i)
\neq 0
\]
for every $\bm{\eta} = (\eta_{1}, \eta_{2})\in(\mathbb{F}_q^*)^{2}$.
The conclusion now follows immediately from
Theorem~\ref{thm: double twist}(b).
\flushright\qed
\end{proof}
\begin{example}
     Let $q=7, r=3, k=5,\bm{h}=(0,0),$ and let $\bm{\eta} = (3, 4)\in (\mathbb{F}_7^{*})^{2}$. Consider the finite field $\mathbb{F}_{343}:=\mathbb{F}_7[\gamma],$ where $\gamma^3+6\gamma^2 = 3.$ Then 
    $$\mathcal{T}_{7, 3}:= 
\{
\gamma^{113}, \gamma^{243}, \gamma^{336}, \gamma^{219}, \gamma^{170},
\gamma^{179}, \gamma^{38}, \gamma^{95}, \gamma^{56}, \gamma^{19},
\gamma^{304}, \gamma^{278}, \gamma^{51}, \gamma^{335},
\gamma^{222}, \gamma^{108}
\}.$$
    Using \textsc{Magma} \cite{Magma}, one verifies that the code $\mathcal{C}_{16,5}(\mathcal{T}_{7,3}, \bm{t}, \bm{h}, \bm{\eta})$ is an $(16, 7^5, 15)$ additive MDS code over $\mathbb{F}_{343}$, in accordance with Proposition~\ref{prop: existence1}.
\end{example}
\begin{proposition}\label{prop: existence2}
    Let $q>3$ and suppose that
$k\equiv -1\pmod r$. Then the code $\mathcal{C}_{|\mathcal{S}_{q,r}|,k}
(\mathcal{S}_{q,r},\bm{t},\bm{h},\bm{\eta})$ is an additive MDS code over $\mathbb{F}_{q^r}$ for every
$\bm{\eta}=(\eta_1, \eta_{2})\in\mathbb{F}_q^*\times (\mathbb{F}_q^*\setminus\{\pm1\})$.
\end{proposition}
\begin{proof}
    Set $\bm{\alpha}:=\mathcal{S}_{q,r}$ and
$n:=|\mathcal{S}_{q,r}|$. Since
$\Nrqr(\alpha_i)=1$ for every $\alpha_i\in\mathcal{S}_{q,r}$, it follows that
\[
(-1)^{k+1} \prod_{j\in J}\Nrqr(\alpha_j)=\pm1,
\]
for every subset $J\subseteq[n]$ with
$|J|=\frac{k+1}{r}$.
If $\eta_{2}\in\mathbb{F}_q^{*}\setminus\{\pm1\},$ then the second condition in
Theorem~\ref{thm: double twist}(b) (that is, $\eta_{2}^{-1}=
(-1)^{k+1} \prod_{j\in J}\Nrqr(\alpha_j)$, for some subset $J\subseteq[n]$ with
$|J|=\frac{k+1}{r}$) fails. Thus,
$\mathcal{C}_{n,k}(\bm{\alpha},\bm{t},\bm{h},\bm{\eta})$ is not an additive almost MDS code; hence, it is an additive MDS code. \qed
\end{proof}
\begin{example}
    Let $q=7, r=3, k=5,\bm{h}=(0,0),$ and let $\bm{\eta} = (3, 4)\in (\mathbb{F}_7^{*})^{2}$. Consider the finite field $\mathbb{F}_{343}:=\mathbb{F}_7[\gamma],$ where $\gamma^3+6\gamma^2 = 3.$ Then 
    $$\mathcal{S}_{7, 3}:= 
\{
\gamma^{6}, \gamma^{12}, \gamma^{18}, \gamma^{24}, \gamma^{30}, \gamma^{36},
\gamma^{48}, \gamma^{60}, \gamma^{66}, \gamma^{72}, \gamma^{90}, \gamma^{96},
\gamma^{132}, \gamma^{138}, \gamma^{144}, \gamma^{174}, \gamma^{180},
\gamma^{186}
\}.
$$
    Using \textsc{Magma} \cite{Magma}, one verifies that the code $\mathcal{C}_{18,5}(\mathcal{S}_{7,3}, \bm{t}, \bm{h}, \bm{\eta})$ is an $(18, 7^5, 17)$ additive MDS code over $\mathbb{F}_{343}$, in accordance with Proposition~\ref{prop: existence2}.
\end{example}
The following remark provides an existence result for additive MDS codes when $k\equiv 0\pmod r$ by Theorem~\ref{thm: double twist}(a).
\begin{remark}
Assume that $k\equiv 0 \pmod r$.
Let
    \[
\mathcal{S}_{n,k,r}(\eta_{2})
=
\left\{
(-1)^{k}
\left(\prod_{i\in I}\Nrqr(\alpha_{i})\right)^{-1}
-
\eta_{2}\sum_{i\in I}\Trqr(\alpha_{i})
\,:\,
I\subseteq[n],
\ |I|=\frac{k}{r}
\right\}.
\]
Then $|\mathcal{S}_{n,k,r}(\eta_{2})|\le\binom{n}{k/r}$. Hence, if $\binom{n}{k/r}<q-1,$ then $\mathbb{F}_q^{*}\setminus \mathcal{S}_{n,k,r}(\eta_{2})\ne \emptyset.$ Therefore, there exists $\eta_{1}\in \mathbb{F}_q^{*}\setminus \mathcal{S}_{n,k,r}(\eta_{2})$, so that, by Theorem~\ref{thm: double twist}(a), the code $\mathcal{C}_{n,k}(\bm{\alpha}, \bm{t}, \bm{h}, \bm{\eta})$ is an additive MDS code. 
\end{remark}

\begin{theorem}
    With the notations as above, if $3<k<\frac{nr}{2},$ then the code $\mathcal{C}_{n,k}(\bm{\alpha}, \bm{t}, \bm{h}, \bm{\eta})$ is not monomially equivalent to an additive RS code.
\end{theorem}
\begin{proof}
    It suffices to show that $\dim_{\mathbb{F}_q}\! \mathcal{C}_{n,k}(\bm{\alpha}, \bm{t}, \bm{h}, \bm{\eta})^{\star 2}\ge 2k.$ The desired conclusion then follows from Lemma \ref{lem: equiv implies dim of schur square equal} and Lemma \ref{lem: dim of schur of ars}.
    
    For each $i\ge 0,$ let $\bm{\alpha}^i:=(\alpha_1^i, \alpha_2^i, \dots, \alpha_n^i)$. A generator matrix of $\mathcal{C}_{n,k}(\bm{\alpha}, 2, h, \eta)$ is
    $$
        G=
        \begin{pmatrix}
            1+\eta_{1} \alpha_{1}^{k} + \eta_{2} \alpha_{1}^{k+1} & \cdots & 1+\eta_{1} \alpha_{n}^{k} + \eta_{2} \alpha_{n}^{k+1} \\
            \alpha_{1} & \cdots & \alpha_{n} \\
            \vdots & \cdots & \vdots \\
            \alpha_{1}^{k-1} & \cdots & \alpha_{n}^{k-1}
        \end{pmatrix}
        \quad = \quad
        \left(
\begin{tabular}{c}
     $\bm{g}_{1}$ \\
     $\bm{g}_{2}$ \\
     $\vdots$ \\
     $\bm{g}_{k}$ 
\end{tabular}
\right),
    $$
    where $\bm{g}_{1}=\bm{1}+\eta_{1} \bm{\alpha}^{k}+\eta_{2}\bm{\alpha}^{k+1},$ and
    $\bm{g}_i=\bm{\alpha}^{i-1},$  for $2\le i\le k$. 
    
    Let 
    \begin{equation*}
        \mathcal{D}:=\{\bm{g}_2\star \bm{g}_2,\ \bm{g}_2\star \bm{g}_3, \dots, \bm{g}_2\star \bm{g}_{k}\} \cup \{ \bm{g}_3\star \bm{g}_{k}, \dots, \bm{g}_{k}\star \bm{g}_{k}\}\cup \{ \bm{g}_{1}\star \bm{g}_{2}, \bm{g}_{1}\star\bm{g}_{k-1}, \bm{g}_{1}\star\bm{g}_{k}\} \subset \mathcal{C}_{n,k}(\bm{\alpha}, \bm{t}, \bm{h}, \bm{\eta})^{\star 2}.
    \end{equation*} 
    Observe that
    \[
    \begin{aligned}
\mathcal{D}
={}&
\{\bm{\alpha}^{2},\bm{\alpha}^{3},\ldots,\bm{\alpha}^{k}\}
\cup
\{\bm{\alpha}^{k+1},\ldots, \bm{\alpha}^{2k-2}\}
\\
&\cup
\{\bm{\alpha}^{1}+\eta_{1}\bm{\alpha}^{k+1}+\eta_{2}\bm{\alpha}_{1}^{k+2}, \bm{\alpha}^{k-2}+\eta_{1}\bm{\alpha}^{2k-2}+\eta_{2}\bm{\alpha}_{1}^{2k-1}, \bm{\alpha}^{k-1}+\eta_{1}\bm{\alpha}^{2k-1}+\eta_{2}\bm{\alpha}_{1}^{2k}\}.
\end{aligned}
    \]
    One can easily verify that $\mathcal{D}$ is linearly independent over $\mathbb{F}_{q}$. Since $|\mathcal{D}| = 2k$, it follows that $
\dim_{\mathbb{F}_q}
\bigl(\mathcal{C}_{n,k}(\bm{\alpha},\bm{t},\bm{h},\bm{\eta})^{\star2}\bigr)
\ge 2k.$ \qed
\end{proof}

\section{Parity-Check Matrices of
$\mathcal{C}_{n,k}(\texorpdfstring{\bm{\alpha}}{\alpha},2,h,\eta)$
and
$\mathcal{C}_{n,k}(\texorpdfstring{\bm{\alpha}}{\alpha} ,(1,2),(0,0),\texorpdfstring{\bm{\eta}}{\eta})$}\label{Section 5}
In this section, we determine parity-check matrices for the codes
$\mathcal{C}_{n,k}(\bm{\alpha},1,h,\eta)$ and
$\mathcal{C}_{n,k}(\bm{\alpha},(1,2),$ $(0,0),\bm{\eta})$.
We begin by fixing some notation that will be used throughout this section.
\begin{table}[H]
\centering
\begin{tabular}{l|p{10cm}}
\hline
\textbf{Notation} & \textbf{Description} \\
\hline
$I_n$
& Identity matrix of order $n$. \\ \hline

$P_{i,j}$
& Permutation matrix obtained from $I_n$ by interchanging columns $i$ and $j$. \\ \hline

$R_{i,j}(\lambda)$
& Elementary matrix obtained from $I_n$ by adding $\lambda$ times row $i$ to row $j$. \\ \hline

$C_{i,j}(\lambda)$
& Elementary matrix obtained from $I_n$ by adding $\lambda$ times column $i$ to column $j$. \\
\hline
\end{tabular}
\end{table}

We recall the following result from \cite{sharma2024mds}.
\begin{lemma}\label{lem:existence of z}
    Let 
    \[
    \mathcal{A} = \left(
    \begin{tabular}{c c c c}
         $1$ & $1$ & $\cdots$ & $1$  \\
         $\alpha_{1}$ & $\alpha_{2}$ & $\cdots$ & $\alpha_{n}$ \\
         $\vdots$ & $\vdots$ & $\cdots$ & $\vdots$ \\
         $\alpha_{1}^{nr-2}$ & $\alpha_{2}^{nr-2}$ & $\cdots$ & $\alpha_{n}^{nr-2}$
    \end{tabular}
    \right).
    \]
    Then there exists $\bm{z} =(z_{1},z_{2},\ldots, z_{n})\in (\mathbb{F}_{q^r}^{*})^{n}$ such that $\Trqr(\mathcal{A}\mathbf{z}^{T}) = \mathbf{0}.$
\end{lemma}
Let $\bm{z}=(z_1,z_2,\ldots,z_n)$ be as in Lemma \ref{lem:existence of z}, and define
    \[
    \mathcal{G} = 
    \left( 
    \begin{tabular}{c c c c}
        $1$ & $1$ & $\cdots$ & $1$ \\
        $\alpha_1$ & $\alpha_2$ & $\ldots$ & $\alpha_n$ \\
        $\vdots$ & $\vdots$ & $\ddots$ & $\vdots$ \\
        $\alpha_{1}^{n-1}$ & $\alpha_{2}^{n-1}$ & $\cdots$ & $\alpha_{n}^{n-1}$\\
    \end{tabular}
    \right)
    \]
    and 
    \[
    \mathcal{H} = 
    \left(
    \begin{tabular}{c c c c}
      $z_{1}\alpha_{1}^{nr-1}$ & $z_{1}\alpha_{1}^{nr-2}$ & $\cdots$ & $z_{1}$\\
        $z_{2}\alpha_{1}^{nr-1}$ & $z_{2}\alpha_{1}^{nr-2}$ & $\cdots$ & $z_{2}$\\
        $\vdots$ & $\vdots$ & $\cdots$ & $\vdots$ \\
         $z_{n}\alpha_{1}^{nr-1}$ & $z_{n}\alpha_{1}^{nr-2}$ & $\cdots$ & $z_{n}$\\
    \end{tabular}
    \right).
    \]
Then 
\begin{equation}\label{eq: GH}
    W = \mathcal{G}\mathcal{H} = 
\left(
\begin{tabular}{c c c c c c c}
     $w_0$ & $0$ & $\cdots$ & $0$ & $0$ & $\cdots$ & $0$ \\
     $w_1$ & $w_0$ & $\cdots$ & $0$ & $0$ & $\cdots$ & $0$ \\
     $\vdots$ & $\vdots$ & $\ddots$ & $\vdots$ & $\vdots$ & $\cdots$ & $\vdots$ \\
     $w_{n-1}$ & $w_{n-2}$ & $\cdots$ & $w_{0}$ & $0$ & $\cdots$ & $0$
\end{tabular}
\right),
\end{equation}
where
\[
w_i=\sum_{j=1}^n\Trqr\!\left(z_j\alpha_j^{\,nr-1+i}\right),
\qquad
0\le i\le n-1,
\]
and the first $n$ columns of $W$ form a lower-triangular Toeplitz matrix, and the remaining $nr-n$ columns are zero.
\subsection{Parity-Check Matrix of
$\mathcal{C}_{n,k}(\texorpdfstring{\bm{\alpha}}{\alpha},2,h,\eta)$}

Pre-multiplying the matrix $W$ in \eqref{eq: GH} by the elementary matrix $R_{k+2,h+1}(\eta)$, we obtain
\[
\begin{pNiceMatrix}[
    first-row,
    last-col
]
  1 &  & h+1 & h+2 &  & k & k+1 & k+2 &  & nr \\[6pt]
w_{0} & \cdots & 0 & 0 & \cdots & 0 & {\color{blue}0} & {\color{blue}0} & {\color{blue}\cdots} & {\color{blue}0} & \hspace{8pt} 1\\ 
\vdots & \vdots & \vdots & \vdots & \vdots & \vdots & {\color{blue}\vdots} & {\color{blue}\vdots} & {\color{blue}\vdots} & {\color{blue}\vdots} & \\
w_{h}+\eta w_{k+1}
& \cdots
& w_{0}+\eta w_{k-h+1}
& \eta w_{k-h}
& \cdots
& \eta w_{2}
& {\color{blue}\eta w_{1}}
& {\color{blue}\eta w_{0}}
& {\color{blue}\cdots}
& {\color{blue}0}
& \hspace{8pt} h+1\\
w_{h+1}
& \cdots
& w_{1}
& w_{0}
& \cdots
& 0
& {\color{blue}0}
& {\color{blue}0}
& {\color{blue}\cdots}
& {\color{blue}0}
& \hspace{8pt} h+2\\
\vdots & \vdots & \vdots & \vdots & \vdots & \vdots & {\color{blue}\vdots} & {\color{blue}\vdots} & {\color{blue}\vdots} & {\color{blue}\vdots} & \\
w_{k-1}
& \cdots
& w_{k-h-1}
& w_{k-h-2}
& \cdots
& w_{0}
& {\color{blue}0}
& {\color{blue}0}
& {\color{blue}\cdots}
& {\color{blue}0}
& \hspace{8pt} k\\
w_{k}
& \cdots
& w_{k-h}
& w_{k-h-1}
& \cdots
& w_{1}
& w_{0}
& 0
& \cdots
& 0
& \hspace{8pt} k+1\\
w_{k+1}
& \cdots
& w_{k-h+1}
& w_{k-h}
& \cdots
& w_{2}
& w_{1}
& w_{0}
& \cdots
& 0
& \hspace{8pt} k+2\\
\vdots & \vdots & \vdots & \vdots & \vdots & \vdots & \vdots & \vdots & \vdots & \vdots & \\
w_{n-1}
& \cdots
& w_{n-h-1}
& w_{n-h-2}
& \cdots
& w_{n-k}
& w_{n-k-1}
& w_{n-k-2}
& \cdots
& 0
& \hspace{8pt} n
\end{pNiceMatrix}
\]
We now perform a sequence of elementary column operations to transform the blue submatrix to the zero matrix. We consider the following two cases.
\begin{enumerate}
    \item[(1)] Suppose that $w_0\neq0$. Define the following sequence of elementary column operations:
\begin{eqnarray*}
    Q^{(1)} &=& C_{k+2, h+1}\left(-\frac{w_0 + \eta w_{k-h+1}}{\eta w_0}\right) \prod_{j=h+2}^{k+1} C_{k+2, j}\left(-\frac{w_{k-j+2}}{w_0}\right) \\
    Q^{(2)} &=& \prod_{j=h+1}^{k-1} C_{k, j} \left( -\frac{w_{k-j}}{w_0} \right) \\
    Q^{(3)} &=& \prod_{j=h+1}^{k-2} C_{k-1, j} \left( -\frac{w_{k-j-1}}{w_0} \right) \\
    &\vdots & \\
    Q^{(k-h)} &=& C_{h+2, h+1}\left(-\frac{w_1}{w_0}\right) 
\end{eqnarray*}
Then the matrix $X=R_{k+2, h+1}(\eta) W Q^{(1)}Q^{(2)}\cdots Q^{(k-h)} P_{k+2, h+1}$ has the following form
\[
\begin{pNiceMatrix}[
    first-row,
    last-col
]
  1 &  & h+1 & h+2 &  & k & k+1 & k+2 &  & nr \\[6pt]
* & \cdots & * & * & \cdots & * & 0 & 0 & \cdots & 0 & \hspace{8pt} 1\\ 
\vdots & \vdots & \vdots & \vdots & \vdots & \vdots & \vdots & \vdots & \vdots & \vdots & \\
*
& \cdots
& *
& *
& \cdots
& *
& 0
& 0
& \cdots
& 0
& \hspace{8pt} h+1\\
*
& \cdots
& *
& *
& \cdots
& *
& 0
& 0
& \cdots
& 0
& \hspace{8pt} h+2\\
\vdots & \vdots & \vdots & \vdots & \vdots & \vdots & \vdots & \vdots & \vdots & \vdots & \\
*
& \cdots
& *
& *
& \cdots
& *
& 0
& 0
& \cdots
& 0
& \hspace{8pt} k\\
*
& \cdots
& *
& *
& \cdots
& *
& *
& *
& \cdots
& *
& \hspace{8pt} k+1\\
*
& \cdots
& *
& *
& \cdots
& *
& *
& *
& \cdots
& *
& \hspace{8pt} k+2\\
\vdots & \vdots & \vdots & \vdots & \vdots & \vdots & \vdots & \vdots & \vdots & \vdots & \\
*
& \cdots
& *
& *
& \cdots
& *
& *
& *
& \cdots
& *
& \hspace{8pt} n
\end{pNiceMatrix}
\]
Now, let $\mathcal{H}' = \mathcal{H}Q^{(1)}Q^{(2)}\cdots Q^{(k-h)} P_{k+2, h+1} = (\bm{c}_1\ \bm{c}_2\ \cdots\ \bm{c}_{nr}),$ where $\bm{c}_i$ denotes the $i$th column of $\mathcal{H}'$. It follows that a parity-check matrix of
$\mathcal{C}_{n,k}(\bm{\alpha},2,h,\eta)$ is
\[
H =
\left(
\begin{tabular}{c}
     $\bm{c}_{nr}^{T}$  \\ 
     $\bm{c}_{nr-1}^{T}$  \\ 
     $\vdots$ \\
     $\bm{c}_{k+1}^{T}$  \\ 
\end{tabular}
\right).
\]
\item[(2)] Suppose that $w_{0} = 0$. Let 
\[
m(x) = \prod_{i=1}^{n}\prod_{j=0}^{r-1}(x-\alpha_{i}^{q^j}) = x^{nr}+\sum\limits_{s=0}^{nr-1}m_{s}x^{s}.
\]
Since each $\alpha_i$ has degree $r$ over $\mathbb{F}_q$, the polynomial
$\prod_{j=0}^{r-1}(x-\alpha_i^{q^j})$ is the minimal polynomial of
$\alpha_i$ over $\mathbb{F}_q$. Hence $m(x)\in\mathbb{F}_q[x]$, and
therefore $m_s\in\mathbb{F}_q$ for all $0\le s\le nr-1$.
Since $m(\alpha_{i}) = 0$, for each $1\le i\le n,$
\[
\alpha_{i}^{nr} = -\sum\limits_{s=0}^{nr-1}m_{s}\alpha_{i}^{s}.
\] 
Hence, 
\begin{eqnarray*}
    w_{1} &=& \sum\limits_{i=1}^{n}\Trqr(z_{i}\alpha_{i}^{nr})\\
    &=& \sum\limits_{i=1}^{n}\Trqr\left(z_{i}\left(-\sum\limits_{s=0}^{nr-1}m_{s}\alpha_{i}^{s}\right)\right) \\
    &=& -\sum\limits_{s=0}^{nr-1}m_{s}\sum\limits_{i=1}^{n}\Trqr\left(z_{i}\alpha_{i}^{s}\right) \\
    &=& -m_{nr-1} \sum\limits_{i=1}^{n}\Trqr(z_{i}\alpha_{i}^{nr-1}) \\
    &=& w_{0} \sum\limits_{i=1}^{n}\Trqr(\alpha_{i}).\\
\end{eqnarray*}
Therefore, $w_{1} = 0$ and consequently, in this case,
\[
\mathcal{C}_{n,k}(\bm{\alpha}, 2, h, \eta)^{\perp} = \agrs_{n, nr-k}(\bm{\alpha}, \bm{z}).
\]
This leads to the following theorem.
\begin{theorem}\label{thm:parity-check-one-twist}
With the above notations, consider the code $\mathcal{C}_{n,k}(\bm{\alpha},2,h,\eta)$ and let $w_0=\sum_{i=1}^n\Trqr\!\left(z_i\alpha_i^{\,nr-1}\right)$.
\begin{enumerate}
    \item[(a)] If $w_0\neq 0$, then a parity-check matrix of
    $\mathcal{C}_{n,k}(\bm{\alpha},2,h,\eta)$ is given by
    \[
    H=
    \begin{pmatrix}
    \bm{c}_{nr}^{T}\\
    \bm{c}_{nr-1}^{T}\\
    \vdots\\
    \bm{c}_{k+1}^{T}
    \end{pmatrix},
    \]
    where
    \[
    \mathcal{H}'
    =
    \mathcal{H}
    Q^{(1)}Q^{(2)}\cdots Q^{(k-h)}
    P_{k+2,h+1}
    =
    (\bm{c}_1\ \bm{c}_2\ \cdots\ \bm{c}_{nr}),
    \]
    and the matrices $Q^{(1)},\dots,Q^{(k-h)}$ are defined as above.

    \item[(b)] If $w_0=0$, then
    \[
    \mathcal{C}_{n,k}(\bm{\alpha},2,h,\eta)^{\perp}
    =
    \agrs_{n,nr-k}(\bm{\alpha},\bm{z}).
    \]
\end{enumerate}
\end{theorem}
\begin{proof}
A generator matrix of $\mathcal{C}_{n,k}(\bm{\alpha}, 2, h, \eta)$ is
    $$
        G=\begin{pmatrix}
            1 & \cdots & 1 \\
            \alpha_{1} & \cdots & \alpha_{n} \\
            \vdots & \cdots & \vdots \\
            \alpha_{1}^{h-1} & \cdots & \alpha_{n}^{h-1} \\
            \alpha_{1}^{h}+\eta \alpha_{1}^{k+1} & \cdots & \alpha_{n}^{h}+\eta \alpha_{n}^{k+1} \\
            \alpha_{1}^{h+1} & \cdots & \alpha_{n}^{h+1} \\
            \vdots & \cdots & \vdots \\
            \alpha_{1}^{k-1} & \cdots & \alpha_{n}^{k-1}
        \end{pmatrix}.
        $$
 Observe that $G$ is obtained from
$R_{k+2,h+1}(\eta)\mathcal{G}$ by retaining its first $k$ rows. Consequently, the matrix $\Trqr(GH^{t})$ is precisely the submatrix of \\ $X=R_{k+2,h+1}(\eta)WQ^{(1)}Q^{(2)}\cdots Q^{(k-h)}P_{k+2,h+1}$ formed by its first $k$ rows and last $nr-k$ columns. Since this block of $X$ is zero, it follows that
\[
\Trqr(GH^{T})=O.
\]
Furthermore, the rows of $H$ are linearly independent over
$\mathbb{F}_{q}$. Therefore, the result follows. \qed
\end{proof}
\end{enumerate}
\begin{example}\label{eg: parity_check}
    Let $q=4, r=4, n=4, k=8$ and $h=7$. Let $b$ be a root of the irreducible polynomial $x^2+x+1$ over $\mathbb{F}_{2}$, so that $\mathbb{F}_{4} = \mathbb{F}_{2}(b)$. Further, let $a$ be a root of the irreducible polynomial $y^4 + y^3 + by^2 + by + b$ over $\mathbb{F}_{4}$, so that $\mathbb{F}_{256} = \mathbb{F}_{4}(a)$. Let $\eta = b$ and define
    \[
    \bm{\alpha} := (\alpha_{1}, \alpha_{2}, \alpha_{3}, \alpha_{4}) =  (a,  a^2, a^3, a^5).
    \] 
    Then the generator matrix of $\mathcal{C}_{n,k}(\bm{\alpha}, 2, 7, \eta)$ is 
    \[
G = 
\begin{pmatrix}
1 & 1 & 1 & 1 \\
\alpha_1 & \alpha_2 & \alpha_3 & \alpha_4 \\
\alpha_1^2 & \alpha_2^2 & \alpha_3^2 & \alpha_4^2 \\
\alpha_1^3 & \alpha_2^3 & \alpha_3^3 & \alpha_4^3 \\
\alpha_1^4 & \alpha_2^4 & \alpha_3^4 & \alpha_4^4 \\
\alpha_1^5 & \alpha_2^5 & \alpha_3^5 & \alpha_4^5 \\
\alpha_1^6 & \alpha_2^6 & \alpha_3^6 & \alpha_4^6 \\
\alpha_1^7 + \eta \alpha_{1}^{9}& \alpha_2^7 + \eta \alpha_{2}^{9} & \alpha_3^7 + \eta \alpha_{3}^{9}& \alpha_4^7 + \eta \alpha_{4}^{9}
\end{pmatrix} \quad = \quad
\begin{pmatrix}
1      & 1      & 1      & 1 \\
a      & a^2    & a^3    & a^5 \\
a^2    & a^4    & a^6    & a^{10} \\
a^3    & a^6    & a^9    & a^{15} \\
a^4    & a^8    & a^{12} & a^{20} \\
a^5    & a^{10} & a^{15} & a^{25} \\
a^6    & a^{12} & a^{18} & a^{30} \\
a^{174} & a^{217} & a^{230} & a^{211}
\end{pmatrix}.
\]
    By Lemma \ref{lem:existence of z}, we get $\bm{z} = (a^2, a^{75}, a^{124}, a^{24}).$ Moreover, $w_{0} = b, w_1 = b^2$ and $w_2 = b$. Applying Theorem \ref{thm:parity-check-one-twist}, we obtain the following parity-check matrix
    \[
\resizebox{\textwidth}{!}{$
H =
\begin{pmatrix}
z_1 & z_2 & z_3 & z_4 \\
z_1\alpha_1 & z_2\alpha_2 & z_3\alpha_3 & z_4\alpha_4 \\
z_1\alpha_1^2 & z_2\alpha_2^2 & z_3\alpha_3^2 & z_4\alpha_4^2 \\
z_1\alpha_1^3 & z_2\alpha_2^3 & z_3\alpha_3^3 & z_4\alpha_4^3 \\
z_1\alpha_1^4 & z_2\alpha_2^4 & z_3\alpha_3^4 & z_4\alpha_4^4 \\
z_1\alpha_1^5 & z_2\alpha_2^5 & z_3\alpha_3^5 & z_4\alpha_4^5 \\
z_1\!\left(\alpha_1^8-\frac{w_0+\eta w_2}{\eta w_0}\alpha_1^6\right) &
z_2\!\left(\alpha_2^8-\frac{w_0+\eta w_2}{\eta w_0}\alpha_2^6\right) &
z_3\!\left(\alpha_3^8-\frac{w_0+\eta w_2}{\eta w_0}\alpha_3^6\right) &
z_4\!\left(\alpha_4^8-\frac{w_0+\eta w_2}{\eta w_0}\alpha_4^6\right) \\
z_1\!\left(\alpha_1^7-\frac{w_1}{w_0}\alpha_1^6\right) &
z_2\!\left(\alpha_2^7-\frac{w_1}{w_0}\alpha_2^6\right) &
z_3\!\left(\alpha_3^7-\frac{w_1}{w_0}\alpha_3^6\right) &
z_4\!\left(\alpha_4^7-\frac{w_1}{w_0}\alpha_4^6\right)
\end{pmatrix}
\quad = \quad
\begin{pmatrix}
a^2    & a^{75}  & a^{124} & a^{24} \\
a^3    & a^{77}  & a^{127} & a^{29} \\
a^4    & a^{79}  & a^{130} & a^{34} \\
a^5    & a^{81}  & a^{133} & a^{39} \\
a^6    & a^{83}  & a^{136} & a^{44} \\
a^7    & a^{85}  & a^{139} & a^{49} \\
a^{67} & a^{251} & a^{67}  & a^8 \\
a^{49} & a^{146} & a^{204} & a^{227}
\end{pmatrix}.
$}
\] All the computations were carried out in \textsc{Magma}.
\end{example}
\begin{remark}
Theorem~3.6(a) of \cite{JiayuMa2026} does not appear to be correct. Indeed, the parity-check matrix $H_1$ obtained for the code in Example~\ref{eg: parity_check} by applying \cite[Theorem~3.6(a)]{JiayuMa2026} does not satisfy $\Trqr(GH_1^{T}) \neq O.$ In fact, following the notations of \cite[Theorem~3.6]{JiayuMa2026}, we have $\omega = b, \gamma = b^2$ and $\eta\omega+1 = b$. Consequently, the parity-check matrix $H_1$ is 
\[
\resizebox{\textwidth}{!}{$
H_1 = \begin{pmatrix}
z_1 & z_2 & z_3 & z_4 \\
z_1\alpha_1 & z_2\alpha_2 & z_3\alpha_3 & z_4\alpha_4 \\
z_1\alpha_1^2 & z_2\alpha_2^2 & z_3\alpha_3^2 & z_4\alpha_4^2 \\
z_1\alpha_1^3 & z_2\alpha_2^3 & z_3\alpha_3^3 & z_4\alpha_4^3 \\
z_1\alpha_1^4 & z_2\alpha_2^4 & z_3\alpha_3^4 & z_4\alpha_4^4 \\
z_1\alpha_1^5 & z_2\alpha_2^5 & z_3\alpha_3^5 & z_4\alpha_4^5 \\
z_1\!\left(\alpha_1^6-\frac{1}{\gamma}\alpha_1^8-\frac{\eta}{\eta\omega+1}\alpha_1^9\right) &
z_2\!\left(\alpha_2^6-\frac{1}{\gamma}\alpha_2^8-\frac{\eta}{\eta\omega+1}\alpha_2^9\right) &
z_3\!\left(\alpha_3^6-\frac{1}{\gamma}\alpha_3^8-\frac{\eta}{\eta\omega+1}\alpha_3^9\right) &
z_4\!\left(\alpha_4^6-\frac{1}{\gamma}\alpha_4^8-\frac{\eta}{\eta\omega+1}\alpha_4^9\right) \\
z_1\alpha_1^7 & z_2\alpha_2^7 & z_3\alpha_3^7 & z_4\alpha_4^7
\end{pmatrix}
\quad = \quad
\begin{pmatrix}
a^2    & a^{75}  & a^{124} & a^{24} \\
a^3    & a^{77}  & a^{127} & a^{29} \\
a^4    & a^{79}  & a^{130} & a^{34} \\
a^5    & a^{81}  & a^{133} & a^{39} \\
a^6    & a^{83}  & a^{136} & a^{44} \\
a^7    & a^{85}  & a^{139} & a^{49} \\
a^{129} & a^{142} & a^{159} & a^{177} \\
a^9    & a^{89}  & a^{145} & a^{59}
\end{pmatrix}.
$}
\]
\end{remark}

\subsection{Parity-Check Matrix of
$\mathcal{C}_{n,k}(\texorpdfstring{\bm{\alpha}}{\alpha} ,(1,2),(0,0),\texorpdfstring{\bm{\eta}}{\eta})$}
Pre-multiplying $W$ in \eqref{eq: GH} by $R_{k+1, 1}(\eta_1)$ and $R_{k+2, 1}(\eta_2)$, we get
\[
\begin{pNiceMatrix}[
    first-row,
    last-col
]
  1 & 2 &  & k & k+1 & k+2 &  & nr \\[6pt]
 w_{0} + \eta_1w_k + \eta_2w_{k+1} & \eta_1w_{k-1} + \eta_2w_{k} & \cdots & \eta_1w_1 + \eta_2w_{2} & {\color{blue}\eta_1w_0 + \eta_2w_{1}} & {\color{blue} \eta_2w_{0}} & {\color{blue}\cdots} & {\color{blue}0} & \hspace{8pt} 1\\ 

 w_{1}
& w_{0}
& \cdots
& 0
& {\color{blue}0}
& {\color{blue}0}
& {\color{blue}\cdots}
& {\color{blue}0}
& \hspace{8pt} 2\\
 \vdots & \vdots & \vdots & \vdots & {\color{blue}\vdots} & {\color{blue}\vdots} & {\color{blue}\vdots} & {\color{blue}\vdots} & \\
 w_{k-1}
& w_{k-2}
& \cdots
& w_{0}
& {\color{blue}0}
& {\color{blue}0}
& {\color{blue}\cdots}
& {\color{blue}0}
& \hspace{8pt} k\\

 w_{k}
& w_{k-1}
& \cdots
& w_{1}
& w_{0}
& 0
& \cdots
& 0
& \hspace{8pt} k+1\\
 w_{k+1}
& w_{k}
& \cdots
& w_{2}
& w_{1}
& w_{0}
& \cdots
& 0
& \hspace{8pt} k+2\\
 \vdots & \vdots & \vdots & \vdots & \vdots & \vdots & \vdots & \vdots & \\

 w_{n-1}
& w_{n-2}
& \cdots
& w_{n-k}
& w_{n-k-1}
& w_{n-k-2}
& \cdots
& 0
& \hspace{8pt} n
\end{pNiceMatrix}
\]
We consider the following two cases. 
\begin{enumerate}
    \item[(1)] If $w_0 \neq 0$, then define
\begin{eqnarray*}
    Q^{(1)} &=& C_{k+2, 1}\left(-\frac{w_0 + \eta_{1}w_{k} + \eta_{2}w_{k+1}}{\eta_{2} w_0}\right) \prod_{j=2}^{k+1} C_{k+2, j}\left(-\frac{\eta_{1}w_{k+1-j}+\eta_{2}w_{k+2-j}}{\eta_{2}w_0}\right) \\
    Q^{(2)} &=& \prod_{j=1}^{k-1} C_{k, j} \left( -\frac{w_{k-j}}{w_0} \right) \\
    Q^{(3)} &=& \prod_{j=1}^{k-2} C_{k-1, j} \left( -\frac{w_{k-j-1}}{w_0} \right) \\
    &\vdots & \\
    Q^{(k)} &=& C_{2, 1}\left(-\frac{w_1}{w_0}\right).
\end{eqnarray*}
Then $R_{k+1, 1}(\eta_1) R_{k+2, 1}(\eta_2) W Q^{(1)}Q^{(2)}\cdots Q^{(k)} P_{k+2, 1}$ is
\[
\begin{pNiceMatrix}[
    first-row,
    last-col
]
  1 & 2 &  & k & k+1 & k+2 &  & nr \\[6pt]
* & * & \cdots & * & 0 & 0 & \cdots & 0& \hspace{8pt} 1\\ 

 *
& *
& \cdots
& *
& 0
& 0
& \cdots
& 0
& \hspace{8pt} 2\\
 \vdots & \vdots & \vdots & \vdots & \vdots & \vdots & \vdots & \vdots & \\
 *
& *
& \cdots
& *
& 0
& 0
& \cdots
& 0
& \hspace{8pt} k\\

 *
& *
& \cdots
& *
& *
& *
& \cdots
& *
& \hspace{8pt} k+1\\
 *
& *
& \cdots
& *
& *
& *
& \cdots
& *
& \hspace{8pt} k+2\\
 \vdots & \vdots & \vdots & \vdots & \vdots & \vdots & \vdots & \vdots & \\

 *
& *
& \cdots
& *
& *
& *
& \cdots
& *
& \hspace{8pt} n
\end{pNiceMatrix}
\]
Now, let $\mathcal{H}' = \mathcal{H}Q^{(1)}Q^{(2)}\cdots Q^{(k)} P_{k+2, 1} = (\bm{c}_1\ \bm{c}_2\ \cdots\ \bm{c}_{nr}).$ Then the parity check matrix is 
\[
H =
\left(
\begin{tabular}{c}
     $\bm{c}_{nr}^{T}$  \\ 
     $\bm{c}_{nr-1}^{T}$  \\ 
     $\vdots$ \\
     $\bm{c}_{k+1}^{T}$  \\ 
\end{tabular}
\right).
\]

\item[(2)] If $w_{0} = 0$, then $w_{1} = 0$. Hence,
\[
\mathcal{C}_{n,k}(\bm{\alpha}, \bm{t}, \bm{h}, \bm{\eta})^{\perp} = \agrs_{n, nr-k}(\bm{\alpha}, \bm{z}).
\]
\end{enumerate}
This leads to the following theorem.
\begin{theorem}\label{thm: parity-check double twist}
With the above notations, consider the code $\mathcal{C}_{n,k}(\bm{\alpha},\bm{t},\bm{h},\bm{\eta})$ and let $w_0=\sum_{i=1}^n\Trqr\!\left(z_i\alpha_i^{\,nr-1}\right)$.
\begin{enumerate}
    \item[(a)] If $w_0\neq 0$, then a parity-check matrix of
    $\mathcal{C}_{n,k}(\bm{\alpha},\bm{t},\bm{h},\bm{\eta})$ is given by
    \[
    H=
    \begin{pmatrix}
    \bm{c}_{nr}^{T}\\
    \bm{c}_{nr-1}^{T}\\
    \vdots\\
    \bm{c}_{k+1}^{T}
    \end{pmatrix},
    \]
    where
    \[
    \mathcal{H}'
    =
    \mathcal{H}
    Q^{(1)}Q^{(2)}\cdots Q^{(k)}
    P_{k+2,1}
    =
    (\bm{c}_1\ \bm{c}_2\ \cdots\ \bm{c}_{nr}),
    \]
    and the matrices $Q^{(1)},\dots,Q^{(k)}$ are defined as above.

    \item[(b)] If $w_0=0$, then
    \[
    \mathcal{C}_{n,k}(\bm{\alpha},\bm{t},\bm{h},\bm{\eta})^{\perp}
    =
    \agrs_{n,nr-k}(\bm{\alpha},\bm{z}).
    \]
\end{enumerate}
\end{theorem}
\begin{example}
     Let $q=4, r=2, n=4, k=2$ and $\bm{h}=(0, 0)$. Let $b$ be a root of the irreducible polynomial $x^2+x+1$ over $\mathbb{F}_{2}$, so that $\mathbb{F}_{4} = \mathbb{F}_{2}(b)$. Further, let $a$ be a root of the irreducible polynomial $y^2 + y + b$ over $\mathbb{F}_{4}$, so that $\mathbb{F}_{16} = \mathbb{F}_{4}(a)$. Let $\bm{\eta} = (b, b^2)$ and define
    \[
    \bm{\alpha} := (\alpha_{1}, \alpha_{2}, \alpha_{3}, \alpha_{4}) =  (a,  a^2, a^3, a^6).
    \] 
    Then the generator matrix of $\mathcal{C}_{n,k}(\bm{\alpha}, \bm{t}, \bm{h}, \bm{\eta})$ is 
    \[
    \resizebox{\textwidth}{!}{$
    G = \begin{pmatrix}
1+\eta_{1}\alpha_{1}^2+\eta_{2}\alpha_{1}^{3} & 1+\eta_{1}\alpha_{2}^2+\eta_{2}\alpha_{2}^{3} & 1+\eta_{1}\alpha_{3}^2+\eta_{2}\alpha_{3}^{3} & 1+\eta_{1}\alpha_{4}^2+\eta_{2}\alpha_{4}^{3} \\
\alpha_{1} & \alpha_{2}  & \alpha_{3} & \alpha_{4}
\end{pmatrix}
\quad = \quad
\begin{pmatrix}
a^{10} & a^{14} & a^{6} & a^{3} \\
a      & a^{2}  & a^{3} & a^{6}
\end{pmatrix}.
$}
\]
By Lemma \ref{lem:existence of z}, we get $\bm{z} = (a^4, a^{13}, a^{7}, a^{4}).$ Moreover, $w_{0} = b^2, w_1 = b^2, w_2 = 0$ and $w_3 = b^2$. Let
\[
\delta = \frac{w_1}{w_0}\frac{\eta_1 w_1 + \eta_2 w_2}{\eta_2 w_0} - \frac{w_0+\eta_1 w_2 + \eta_2 w_3}{\eta_2 w_0}, \quad \delta' = \frac{\eta_1 w_0 + \eta_2 w_1}{\eta_2 w_0}.
\]
Applying Theorem \ref{thm: parity-check double twist}, we obtain the following parity-check matrix:
  \[
\resizebox{\textwidth}{!}{$
H =
\begin{pmatrix}
z_1 & z_2 & z_3 & z_4 \\
z_1\alpha_1 & z_2\alpha_2 & z_3\alpha_3 & z_4\alpha_4 \\
z_1\alpha_1^2 & z_2\alpha_2^2 & z_3\alpha_3^2 & z_4\alpha_4^2 \\
z_1\alpha_1^3 & z_2\alpha_2^3 & z_3\alpha_3^3 & z_4\alpha_4^3 \\
z_1\!\left(\alpha_1^7-\frac{w_1}{w_0}\alpha_{1}^{6} + \delta \alpha_1^4 \right) &
z_2\!\left(\alpha_2^7-\frac{w_1}{w_0}\alpha_{2}^{6} + \delta \alpha_2^4 \right) &
z_3\!\left(\alpha_3^7-\frac{w_1}{w_0}\alpha_{3}^{6} + \delta \alpha_3^4 \right) &
z_4\!\left(\alpha_4^7-\frac{w_1}{w_0}\alpha_{4}^{6} + \delta \alpha_4^4 \right) \\
z_1\!\left(\alpha_1^5-\delta'\alpha_1^4\right) &
z_2\!\left(\alpha_2^5-\delta'\alpha_2^4\right) &
z_3\!\left(\alpha_3^5-\delta'\alpha_3^4\right) &
z_4\!\left(\alpha_4^5-\delta'\alpha_4^4\right)
\end{pmatrix}
\quad = \quad 
\begin{pmatrix}
a^4  & a^{13} & a^7  & a^4 \\
a^5  & 1      & a^{10} & a^{10} \\
a^6  & a^2    & a^{13} & a \\
a^7  & a^4    & a      & a^7 \\
a^{14} & a^3  & a^9    & a^8 \\
a^{10} & a^7  & 1      & a^7
\end{pmatrix}.
$}
\]
\end{example}
All the computations were carried out in \textsc{Magma}.
\section{Conclusion}\label{Section 6}
    In this article, we have investigated two families of additive TRS codes: the first with a single twist $t=2$ and an arbitrary hook, and the second with double twists $\bm{t}=(1,2)$ and hooks $\bm{h}=(0,0)$. We have established a necessary and sufficient condition for each of these codes to be additive MDS. Moreover, by analyzing the Schur squares of these codes, we have obtained several families of additive MDS codes that are not equivalent to additive RS codes. Finally, we have determined parity-check matrices for both families of additive TRS codes.

    In future, additive TRS codes with two arbitrary twists and hooks could be investigated, along with the determination of their corresponding parity-check matrices. Furthermore, extending the additive TGRS framework to the counterparts of other notions of TGRS codes in the literature may yield new classes of additive MDS codes.

\section*{Declarations}
\subsection*{Conflict of Interest}
All authors declare that they have no conflict of interest.

\section*{Acknowledgments}
The second author expresses gratitude to MHRD, India, for financial support in the form of a Senior Research Fellowship at the Indian Institute of Technology Delhi.

\bibliographystyle{abbrv}
\bibliography{subfield}

\end{document}